\documentclass[lettersize,journal]{IEEEtran}

\usepackage[linesnumbered,ruled,vlined]{algorithm2e}
\usepackage{cite}
\usepackage{tabularx}
\usepackage{blindtext}
\usepackage{amsmath}
\usepackage{amsfonts}

\usepackage{pgfplots}   
\usepackage{tikz}       
\usepackage{subcaption} 
\usepackage{graphicx}
\usepackage{textcomp}
\usepackage{color}
\usepackage{xcolor}
\usepackage{multirow}
\usepackage{enumitem}
\usepackage{dsfont}

\usepackage{amssymb}  
\usepackage[hidelinks]{hyperref} 
\usepackage{authblk}
\usepackage{orcidlink}

\usepackage{amsthm} 
\newtheorem{definition}{Definition}

\newtheorem{theorem}{Theorem}
\newtheorem{corollary}{Corollary}

\begin{document}
\newcommand{\kw}[1]{{\ensuremath {\mathsf{#1}}}\xspace}
\newcommand{\stitle}[1]{\vspace{1ex} \noindent{\textbf{#1}}}

\SetKwFunction{FButterflies}{UpdateBC}
\SetKwProg{Fn}{Function}{:}{}

\definecolor{c1}{RGB}{42,99,172} %
\definecolor{c2}{RGB}{255,88,93}
\definecolor{c3}{RGB}{255,181,73}
\definecolor{c4}{RGB}{119,71,64} %
\definecolor{c5}{RGB}{228,123,121} %
\definecolor{c6}{RGB}{208,167,39} %
\definecolor{c7}{RGB}{0,51,153}
\definecolor{c8}{RGB}{56,140,139} 
\definecolor{c9}{RGB}{0,0,0} 

\newcommand{\todo}[1]{\textcolor{red}{$\Rightarrow:$ #1}}
\newcommand{\revise}[1]{#1}

\newcommand{\reviseminor}[1]{#1}

\newcommand{\reviseminorminor}[1]{#1}

\newcommand{\reffig}[1]{Figure~\ref{fig:#1}}
\newcommand{\refsec}[1]{Section~\ref{sec:#1}}
\newcommand{\reftable}[1]{Table~\ref{tab:#1}}
\newcommand{\refalg}[1]{Algorithm~\ref{alg:#1}}
\newcommand{\refeq}[1]{Equation~\ref{eq:#1}}
\newcommand{\refdef}[1]{Definition~\ref{def:#1}}

\newcommand{\twitterut}{\textit{Twitter-ut}\xspace}
\newcommand{\editfrwiki}{\textit{Edit-frwiki}\xspace}
\newcommand{\amazonratings}{\textit{AmazonRatings}\xspace}
\newcommand{\movielens}{\textit{Movie-lens}\xspace}
\newcommand{\edititwiki}{\textit{Edit-itwiki}\xspace}
\newcommand{\lastfmband}{\textit{Lastfm-band}\xspace}
\newcommand{\discogs}{\textit{Discogs}\xspace}
\newcommand{\yahoosongs}{\textit{Yahoo-songs}\xspace}
\newcommand{\stackoverflow}{\textit{StackoverFlow}\xspace}
\newcommand{\livejournal}{\textit{LiveJournal}\xspace}
\newcommand{\deliciousui}{\textit{Delicious-ui}\xspace}
\newcommand{\orkut}{\textit{Orkut}\xspace}

\newcommand{\deabc}{\kw{DEABC}}
\newcommand{\fable}{\kw{FABLE}}
\newcommand{\deabcpro}{\kw{DEABC}}
\newcommand{\abacus}{\kw{ABACUS}}
\newcommand{\fleet}{\kw{FLEET3}}
\newcommand{\cas}{\kw{CAS{-}R}}
\newcommand{\mascot}{\kw{MASCOT}}
\newcommand{\triest}{\kw{TRI\grave{E}ST}}
\newcommand{\fm}{\kw{FM}}
\newcommand{\kmv}{\kw{KMV}}

\title{Counting Butterflies over Streaming Bipartite Graphs with Duplicate Edges}

\author{Lingkai Meng\orcidlink{0009-0002-7961-9131}, Long Yuan\orcidlink{0000-0001-8111-0401}, Xuemin Lin\orcidlink{0000-0003-2396-7225},~\IEEEmembership{Fellow,~IEEE}
, Chengjie Li\orcidlink{0009-0005-5372-2335}, Kai Wang\orcidlink{0000-0002-3123-2184}, and Wenjie Zhang\orcidlink{0000-0001-6572-2600}
\thanks{* Long Yuan is the corresponding author.}

\thanks{Lingkai Meng, Xuemin Lin, and Chengjie Li are with Antai College of Economics and Management, Shanghai Jiao Tong University, Shanghai 200052, China (mlk123@sjtu.edu.cn, xuemin.lin@gmail.com, lichengjie@sjtu.edu.cn). 

Kai Wang is with the Data-Driven Management Decision Making Lab, Antai College of Economics and Management, Shanghai Jiao Tong University, Shanghai 200052, China (w.kai@sjtu.edu.cn).

Long Yuan is with Nanjing University of Science and Technology, School of Computer Science and Technology, Nanjing 210094, China (longyuan@njust.edu.cn)

Wenjie Zhang is with University of New South Wales, Sydney, Australia (wenjie.zhang@unsw.edu.au)} 
\thanks{Manuscript received April 19, 2021; revised August 16, 2021.}}

\markboth{Journal of \LaTeX\ Class Files,~Vol.~14, No.~8, August~2021}%
{Shell \MakeLowercase{\textit{et al.}}: A Sample Article Using IEEEtran.cls for IEEE Journals}


\maketitle

\begin{abstract}
Bipartite graphs are commonly used to model relationships between two distinct entities in real-world applications, such as user-product interactions, user-movie ratings and collaborations between authors and publications. A butterfly (a 2×2 bi-clique) is a critical substructure in bipartite graphs, playing a significant role in tasks like community detection, fraud detection, and link prediction. As more real-world data is presented in a streaming format, efficiently counting butterflies in streaming bipartite graphs has become increasingly important. However, most existing algorithms typically assume that duplicate edges are absent, which is hard to hold in real-world graph streams, as a result, they tend to sample edges that appear multiple times, leading to inaccurate results. The only algorithm designed to handle duplicate edges is \fable, but it suffers from significant limitations, including high variance, substantial time complexity, and memory inefficiency due to its reliance on a priority queue. To overcome these limitations, we introduce \deabcpro (Duplicate-Edge-Aware Butterfly Counting), an innovative method that uses bucket-based priority sampling to accurately estimate the number of butterflies, accounting for duplicate edges. Compared to existing methods, \deabcpro significantly reduces memory usage by storing only the essential sampled edge data while maintaining high accuracy. We provide rigorous proofs of the unbiasedness and variance bounds for \deabcpro, ensuring they achieve high accuracy. 
We compare \deabcpro with state-of-the-art algorithms on real-world streaming bipartite graphs. The results show that our \deabcpro outperforms existing methods in memory efficiency and accuracy, while also achieving significantly higher throughput.
\end{abstract}


\section{Introduction}
\label{sec:intro}

Butterfly (a $2\times 2$ biclique) is  the most basic subgraph structure in bipartite graphs, such as $\{\{u_1, u_2\}, \{v_1, v_2\}\}$ and $\{\{u_3, u_5\}, \{v_4, v_5\}\}$ in Figure~\ref{fig:butterfly},  that captures the local connectivity patterns between vertices~\cite{DBLP:conf/bigdata/WangFC14}. Counting the number of butterflies in a given bipartite graph is a fundamental problem in bipartite graph analysis and has a wide range of applications.  For example, in bipartite graph mining~\cite{DBLP:journals/compnet/AksoyKP17,lind2005cycles,robins2004small}, butterfly counting is the core ingredient of the well-known measurement 
\emph{bipartite clustering coefficient}, which equals $4 \times c_{\Join} / c_{\ltimes}$,
where  $c_{\Join}$  and  $ c_{\ltimes}$  represent the number of butterflies and three-paths respectively. 
\reviseminorminor{This measurement reveals how tightly entities are grouped, making it particularly valuable in 
community detection~\cite{zhang2008clustering}, representative graph sampling~\cite{zhang2017clustering}, and studies of social collective behaviors~\cite{david2020herding}.}
Moreover, butterfly counting is crucial for identifying \textit{$k$-bitruss} in bipartite graphs, which represents the largest subgraph where each edge is contained in at least $k$ butterflies \cite{DBLP:conf/icde/Wang0Q0020,valejo2018community,abidi2024searching}. This has important applications in areas such as community detection \cite{dong2021butterfly,weng2022distributed,DBLP:journals/vldb/ZhouWC23,DBLP:conf/kdd/Sanei-MehriST18,DBLP:journals/pvldb/XuZYLDDH22}, spam detection~\cite{fang2021cohesive,DBLP:journals/pvldb/XuZYLDDH22,DBLP:journals/pacmmod/WangLLS0023}, fraud detection~\cite{geetha2022hybrid,samanvita2024fraud}. Given its importance, butterfly counting has attracted significant attention recently \cite{DBLP:conf/bigdata/WangFC14,DBLP:conf/kdd/Sanei-MehriST18,DBLP:journals/pvldb/WangLQZZ19,DBLP:books/crc/22/ShiS22,DBLP:conf/ipps/AcostaLP22,DBLP:journals/pacmmod/WangLLS0023,DBLP:journals/vldb/WangLLSTZ24,DBLP:journals/pvldb/XuZYLDDH22,DBLP:journals/vldb/XiaZXZYLDDHM24,DBLP:conf/infocom/Wang00LH24,DBLP:conf/icde/HeW0LNZ24,DBLP:conf/cikm/Sanei-MehriZST19,DBLP:journals/tkde/LiWJZZTYG22,DBLP:journals/tkdd/SheshboloukiO22,DBLP:conf/icde/PapadiasKPQM24,DBLP:journals/pacmmod/ZhangCWYG23,DBLP:journals/corr/abs-2310-11886}.


\begin{figure}[t]\centering
    \scalebox{0.42}[0.42]{\includegraphics{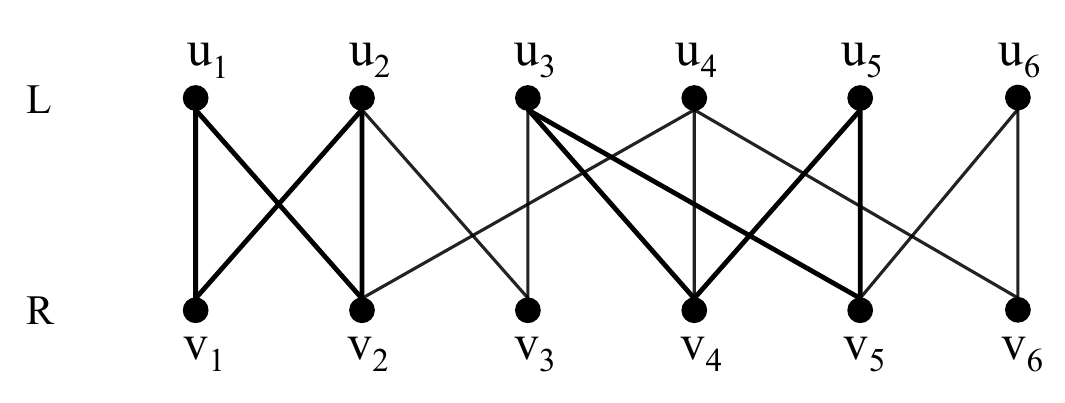}}
    \caption{Butterflies in A Bipartite Graph}
    \label{fig:butterfly}
\end{figure}

\stitle{Motivation.} Real-world bipartite graphs are inherently streaming, with new edges continuously arriving at high speed.   
For example, Alibaba reports that during a peak period in 2017, its customer purchase activities generated 320 PB of log data over just six hours \cite{DBLP:conf/icde/PapadiasKPQM24,DBLP:journals/tkdd/SheshboloukiO22}.  
Given the potentially large or even infinite size of these streams, as well as their high velocity, storing the entire graph in main memory and then counting butterflies afterward is impractical for real-world applications \cite{lin2015scalable,weng2022distributed}. Consequently, a range of memory-efficient one-pass streaming algorithms which can efficiently estimate the number of butterflies in a bipartite graph stream within a user-specified memory limit are proposed in the literature \cite{DBLP:journals/tkde/LiWJZZTYG22,DBLP:conf/icde/PapadiasKPQM24,DBLP:conf/cikm/Sanei-MehriZST19}. However, almost all of existing algorithms assume that edges are unique, an assumption that often fails in real-world bipartite graph stream scenarios \cite{DBLP:journals/pvldb/WangQSZTG17,DBLP:journals/datamine/JungLLK19}. 
\revise{Duplicate edges are commonly present in real-world bipartite graph streams, such as system-level retries due to network instability (e.g., in payment or order systems)~\cite{sun2024fable,tang2016survey,jha2015counting,ives1999adaptive,kandula2008s} and log duplication caused by distributed collection mechanisms in telemetry and monitoring systems~\cite{DBLP:journals/pvldb/WangQSZTG17,DBLP:journals/datamine/JungLLK19}.
However, duplicate edges can lead to incorrect graph analysis if they are not properly handled. 
\reviseminor{Butterfly counting provides the basis for computing classic metrics such as transitivity and clustering coefficient, which are widely used to measure community structure in bipartite networks~\cite{DBLP:conf/kdd/Sanei-MehriST18, DBLP:journals/vldb/WangLQZZ22,DBLP:journals/tkdd/SheshboloukiO22}. These metrics have been applied to tasks such as community detection~\cite{zhang2008clustering}, representative graph sampling~\cite{zhang2017clustering}, and studies of social collective behaviors~\cite{david2020herding}. Besides, clustering coefficient has also been exploited to analyze the predictive performance of deep neural networks~\cite{you2020graph}. However, repeated interactions may cause the same butterfly structure to be counted multiple times, thereby exaggerating the density of communities and distorting cohesiveness measures. For instance, if a single entity forms repeated connections across both sides of the bipartite graph, the resulting butterfly count may misleadingly suggest stronger community structures. Figure~\ref{fig:normal} illustrates that even when a single user $u_1$ generates duplicate edges due to system failures, the clustering coefficient shifts markedly, from 0.675 in (a) to 0.838 in (b).}
Duplicate edges can also obscure anomalies in fraud detection. Dense butterfly patterns are often used to identify coordinated fraudulent behavior~\cite{DBLP:journals/tkde/LiWJZZTYG22,DBLP:conf/cikm/Sanei-MehriZST19,DBLP:journals/pvldb/CaiKWCZLG23,DBLP:conf/icde/PapadiasKPQM24}, but repeated actions by the same user may inflate these structures, thereby concealing true fraud patterns and weakening detection algorithms.
Overall, failure to effectively process duplicate edges introduces significant bias into analysis results. Our experiments in Section~\ref{sec:ee} show that even a duplication rate as low as 1\% can increase relative error by up to an order of magnitude. While one might consider preprocessing to remove duplicates, this is often infeasible in streaming scenarios due to memory limitations and real-time processing demands, as in the case of Alibaba’s large-scale log data.}

\begin{figure}
    \centering

    \vspace{-0.5em}
    
    \begin{minipage}{0.24\textwidth}
        \centering
        \includegraphics[width=0.9\textwidth]{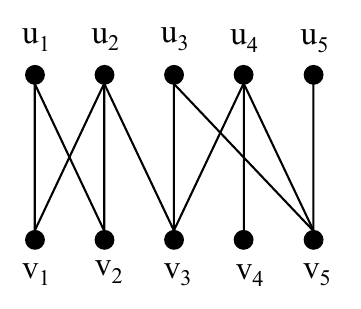}
        \vspace{-0.8em}
        \subcaption{}
    \end{minipage}
    \begin{minipage}{0.24\textwidth}
        \centering
        \includegraphics[width=0.9\textwidth]{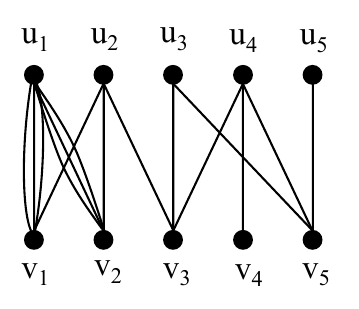}
        \vspace{-0.8em}
        \subcaption{}
    \end{minipage}

  \vspace{-0.5em}
    
    \caption{Impact of Duplicate Edges on Clustering Coefficient ((a) 0.675 vs. (b) 0.838)}
    \vspace{-2em}
    \label{fig:normal}
\end{figure}

Sun et al. propose \fable~\cite{sun2024fable} which is currently the only butterfly counting algorithm that can handle duplicate edges in bipartite graph streaming. However, it still suffers from high
variance and time complexity and requires an additional priority queue to store the priorities of sampled edges. Maintaining this priority queue inevitably increases both time complexity and memory usage.
Motivated by this, in this paper, we study the problem of butterfly counting over streaming bipartite graphs with duplicate edges and aim to devise efficient algorithms that ensure accurate and memory-efficient butterfly counting of large-scale bipartite graph streams.

\stitle{Challenges.} However, it is of great challenge to design such an algorithm for the following reasons:

\begin{itemize}[leftmargin=*]
\item[$\bullet$] \textbf{Memory and computation efficiency.} Considering the large-size and high-speed nature of bipartite graph streams, the designed algorithms must minimize both time complexity and memory consumption while still delivering accurate results. This requirement further exacerbates the complexity of the problem. The state-of-the-art algorithm, \fable, can count butterflies on fixed-size memory, but it requires an additional priority queue to store sampled edges with their priorities, which leads to high memory overhead and time cost.

\item[$\bullet$] \textbf{Ensuring unbiased estimates with low variance.} Achieving unbiased estimates is complicated by the presence of duplicate edges. Without proper sampling that ensures equal probability for each edge, the estimates become skewed. Additionally, variance reflects how much a single estimate value differs from the true value and low variance helps maintain the stability of the results. 
Most existing sampling methods often fail to account for duplicate edges, which compromises the uniformity of the sampling process. This leads to some edges being overrepresented while others are underrepresented, directly affecting the accuracy of butterfly counting. Although \fable ensures unbiasedness, utilizing the \kmv sketch to estimate the number of distinct edges still results in higher variance.

\end{itemize}

\stitle{Our Solution.} To address these challenges, we employ bucket-based priority sampling to maintain a fixed-size edge sample. Using random hashing, duplicate edges are assigned with the same priority, allowing us to effectively handle duplicates without ambiguity (i.e., each edge is sampled with equal probability). Following this idea, we propose memory-efficient \deabcpro (Duplicate-Edge-Aware Butterfly Counting). Specifically, in \deabcpro, we maintain $M$ buckets, where $M$ is the maximum number of edges sampled. Each coming edge is hashed into a fixed number of buckets, and each bucket retains only the edge with the lowest priority. If the bucket is empty or the priority of the new edge is lower than that of the current edge in the bucket, the new edge will replace the existing edge in the bucket. The butterfly count is then updated using a correction factor based on the sampling probability, which is equal to the proportion of sampled edges to all distinct edges. The number of distinct edges is estimated by the Flajolet-Martin sketch. This process ensures that we maintain a uniform random sample of distinct edges. We prove that \deabc provides unbiased estimates with low variance for butterfly counts, even in the presence of duplicate edges. Notably, \deabc avoids the need for extra data structures to track duplicates, due to the hash-based priority sampling.

\revise{While \fm sketch has been widely used for cardinality estimation, our method is the first to integrate it with bucket-based priority sampling for butterfly counting over streaming bipartite graphs with duplicates, completely removing the priority queue structure required by \fable, thereby significantly reducing memory overhead and improving computational efficiency.}

\stitle{Contributions.} We make the following contributions in this paper:
 
\begin{itemize}[leftmargin=*]


\item[$\bullet$] \emph{The new time and memory-efficient algorithms for streaming butterfly counting.} We propose \deabcpro that hashes edges into a fixed number of  buckets, retaining only the highest-priority edge per bucket. \deabcpro only stores the fixed number of edges while still efficiently estimating butterfly counts. Compared to the latest algorithms, \deabcpro shows lower time complexity and less memory usage.

\item[$\bullet$] \emph{Theoretical analysis of the proposed algorithms.} We prove that \deabcpro provides unbiased estimates and establishes upper bounds for its variance. Furthermore, we provide a comprehensive analysis of its time and space complexity. Compared to existing methods, \deabcpro achieves lower variance, faster execution times, and reduced memory consumption.

\item[$\bullet$] \emph{Extensive experimental studies on large real datasets.} We conduct extensive experiments on various real-world streaming bipartite graphs with duplicate edges. The results demonstrate the efficiency, accuracy, and low memory usage of \deabcpro, 
corroborating our theoretical findings. 

\end{itemize}


\section{Problem Definition}
\label{sec:pd}


A bipartite graph  $G = (L, R, E)$  consists of two disjoint vertex sets  $L$  and  $R$  (i.e.,  $L \cap R = \emptyset$ ), and an edge set  $E \subseteq L \times R$. Edges exist only between vertices in $L$  and vertices in $R$; there are no edges within $L$ or within $R$. For a vertex $u \in L$ (respectively, $v \in R$), we denote its set of neighbors as $N[u]$ (respectively, $N[v]$).



\begin{definition}[Butterfly]
\label{def:butterfly}
A butterfly in a bipartite graph  $G = (L, R, E)$ is a  $2 \times 2$  complete bipartite subgraph consisting of vertices $u, w \in L$  and $v, x \in R$, such that all possible edges between these vertices are present, i.e.,
$(u, v),\ (u, x),\ (w, v),\ (w, x) \in E$.
\end{definition}

We use $c_{\Join}$ to denote the count of butterflies in a graph.

\begin{definition}[Bipartite Graph Stream]
\label{def:bipartite_graph_stream}
A bipartite graph stream $\Pi$ is a sequence of edges:
\[
\Pi = \left( e^{(1)},\ e^{(2)},\ \dots,\ e^{(t)},\ \dots \right),
\]
where each $e^{(i)} = \left( u^{(i)}, v^{(i)} \right)$ represents an edge between a vertex $u^{(i)} \in L$ and a vertex $v^{(i)} \in R$ at time $i$. 
The sets $L$ and $R$ would expand over time and there could appear duplicate edges, that is, the same edge $(u, v)$ can appear multiple times in the stream.
\end{definition}

\stitle{Problem Statement.} In this paper, we study the problem of butterfly counting in bipartite graph streams. Specifically, given a bipartite graph stream $G = (L, R, E) = (e^{(1)}, e^{(1)},$ $\dots, e^{(t)})$, we goal is to maintain unbiased butterfly count estimates $\hat{c_{\Join}}$ with low variance while operating under limited memory constraint $M$. In real-world applications, multiple interactions between entities or repeated relationship records often result in the same edge appearing multiple times, known as duplicate edges. Therefore, our objective is to design an algorithm that can accurately count butterflies in bipartite graph streams by considering only the existence of edges, even in the presence of duplicate edges.

For clarity, the notation used throughout this paper is summarized in Table \ref{notation}. To simplify notation, the superscript $(t)$ may be omitted when the context is clear. 

\begin{table}[htbp]
\centering
\caption{Table of Notation}
\small
\begin{tabularx}{0.47\textwidth}{|>{\centering\arraybackslash}c|>{\centering\arraybackslash}X|}
\hline
\textbf{Symbol} & \textbf{Description} \\ \hline
$G = (L, R, E)$ & the bipartite graph \\ \hline
$\Pi$ & bipartite graph stream \\ \hline
$N[u]$ & the set of neighbors of vertex $u$ \\ \hline
$c_{\Join}/\hat{c}_{\Join}$ &  true/estimated number of butterflies\\ \hline
$G_s$ & sampled subgraph \\ \hline
$d_{{max}}$ & the maximum degree of any vertex in the sampled subgraph $G_s$ \\ \hline
$M$ & maximum number of sampled edges \\ \hline
$(u^{(t)}, v^{(t)})$ & an edge $(u, v)$ insertion at time $t$ \\ \hline
$\mathbb{Q}$ & priority queue \\ \hline
$\mathbb{H}_{pri}$ & hash functions that map edges into $Uniform(0, 1)$ \\ \hline
$\mathbb{H}_{pos}$ & hash functions that map edges into integers $1,\cdots,M-1$ uniformly \\ \hline
$h_{max}$ & the maximum priority in the queue $\mathbb{Q}$ \\ \hline
$\theta$ & update factor for butterfly counting \\ \hline
$m^{(t)}$ & true number of edges up to time $t$ \\ \hline
$m_d^{(t)}/\hat{m}_d^{(t)}$ & true/estimated number of distinct edges up to time $t$ \\ \hline
$X_{\Join}$ & number of times butterfly  $\Join$  is counted \\ \hline
$B[i]$ & the $i$-th bucket to store the sampled edges\\ \hline
$|B_{\neq \emptyset}|$ & number of non-empty buckets \\ \hline
\end{tabularx}
\label{notation}
\end{table}

\section{Baseline}
\label{sec:baseline}

In this section, we first review the existing butterfly counting algorithm for the bipartite graph stream with duplicate edges \fable, proposed in~\cite{sun2024fable}, which serves as one of the baseline algorithms in our experimental evaluation. Then we analyze its drawbacks. In addition, we provide an analysis of \fable's time and space complexity for comparison with our algorithm in the following section.

\revise{
The main idea of the \fable algorithm is to maintain a fixed-size sample of edges to estimate the total number of butterflies through hash-based priority sampling.
Each incoming edge is assigned a priority using a hash function, with values drawn from the $Uniform(0, 1)$ distribution.
Duplicate edges are assigned the same priority to eliminate their impact on the sampling probability.
\fable maintains a priority queue of sampled edges up to a predefined size limit.
If the sample size has not reached its maximum, the edge is added to the sample, and the butterfly counts are updated exactly.
If the sample is full, the edge with the highest priority in the sample is evicted when the new edge has a lower priority, and the butterfly counts are updated using a correction factor based on the sampling probability (i.e., the highest priority among sampled edges).
}

\stitle{Algorithm.} For clearness of presentation, we give the pseudocode of \fable algorithm in Algorithm~\ref{algo:fable}, which is equivalent to the original
version in \cite{sun2024fable}. 
\fable maintains a sampled subgraph $G_s$ and estimates the butterfly count as new edges arrive in the stream. Each edge is assigned a priority (line 4), and if the sampled subgraph has not yet reached its maximum size, the edge is added, and the butterfly count is updated accurately (lines 5-9). If $G_s$ is full, \fable compares the priority \( pri \) of the new edge with the current maximum priority \( h_{{max}} \) (lines 10-11). If \( pri < h_{{max}} \) and $(u,v) \notin G_s$, the new edge replaces the edge with the maximum priority in $G_s$ (lines 11-12). The priority queue $\mathbb{Q}$ is updated, and the maximum priority $h_{{max}}$ is recalculated as needed (line 13). When the new edge is sampled \fable utilizes an update factor $\theta$ to accordingly adjust the butterfly count. The update factor is the reciprocal of the probability that all the four edges in a butterfly are selected, i.e., $\theta = \frac{M(M-1)(M-2)(M-3)}{\hat{m}_d(\hat{m}_d-1)(\hat{m}_d-2)(\hat{m}_d-3)}$, where $\hat{m}_d$ represents the estimated number of distinct edges. To estimate the number of distinct edges, \fable introduces the K-Minimum Values sketch (\kmv)~\cite{bar2002counting,beyer2007synopses}, which provides an unbiased estimate with the value of $\hat{m}_d = \frac{M-1}{h_{max}}$, where $h_{max}$ is the $M$-th smallest hash value in $G_S$. The Relative Standard Error of \kmv is bounded by $\frac{1}{\sqrt{M-2}}$.

The procedure to compute the increment in the butterfly count caused by adding a new edge $(u, v)$ is shown in lines~18-23. The butterfly increment \( \hat{c}_{\Join}^+ \) is initialized to zero (line 19), then iterates over the neighbors $w$ of $u$ in  $G_s$, checking for common neighbors $l$ between $w$ and $v$~(lines 20-22). For each butterfly formed, it increments \( \hat{c}_{\Join}^+ \) by the update factor $\theta$~(line~22).

\begin{algorithm}[h]
    \DontPrintSemicolon
    \caption{\fable}
    \label{algo:fable}
    \KwIn{Bipartite graph stream $\Pi$ and maximum number $M$ of edges stored}
    \KwOut{the number of butterflies $\hat{c}_{\Join}$}
    \SetKwComment{comment}{$\triangleright$ }{}

    $G_s \gets \emptyset$; $h_{max} \gets 0$; $\hat{c}_{\Join} \gets 0$;
    
    Initialize the priority queue $\mathbb{Q}$ which stores pairs $(pri, edge)$;

    \For{\textbf{each} edge $(u, v) \in \Pi$}{
             $pri \gets \mathbb{H}_{pri}(u,v)$;\\
             \If{$|G_s| < M$ and $(u,v) \notin G_s$}{
                $\hat{c}_{\Join} \gets \hat{c}_{\Join} +$ \textbf{\FButterflies($(u,v), 1.0$)};\\
                $G_s \gets G_s \cup \{(u, v)\}$;
                $\mathbb{Q}.insert(pri, (u,v))$;\\
                \If{$pri > h_{max}$}{$h_{max} \gets pri$;}
                
             }
             \Else{
                 \If{$pri < h_{max}$ and $(u,v) \notin G_s$}{
                    $(u', v') \gets \mathbb{Q}.top().second$;
                    $G_s \gets G_s /(u', v')$;\\
                    $\mathbb{Q}.pop()$;~$h_{max} \gets \mathbb{Q}.top().first$;\\
                    $\hat{m}_d \gets \frac{M-1}{h_{max}}$;~
                    $\theta \gets \frac{M(M-1)(M-2)(M-3)}{\hat{m}_d(\hat{m}_d-1)(\hat{m}_d-2)(\hat{m}_d-3)}$;\\
                    $\hat{c}_{\Join} \gets \hat{c}_{\Join} +$ \textbf{\FButterflies$((u,v), \theta)$};\\
                    $G_s \gets G_s \cup \{(u, v)\}$;\\
                    
                    $\mathbb{Q}.insert(pri, (u,v))$;\\

                 }
             }   
    }



    \Fn{\FButterflies$((u,v), \theta)$}{
        $\hat{c}_{\Join}^+ \gets 0$;\\
        \For{\textbf{each} $w \in N[u]$}{
            \For{\textbf{each} $l \in N[w] \cap N[v]$}{
                 $\hat{c}_{\Join}^+ \gets \hat{c}_{\Join}^+ + \theta$;
            }    
        }
        \Return $\hat{c}_{\Join}^+$\;
    }
\end{algorithm}

\stitle{Accuracy Analysis.} Sun et al., provided a proof of the unbiasedness of the \fable algorithm in \cite{sun2024fable}, Specifically, $\mathbb{E}[\hat{c}_{\Join}] = c_{\Join}$. \fable also provides estimates of bounded variance, which is defined as:

\begin{equation}
    \mathrm{Var}(\hat{c}_{\Join}) = c_{\Join} \Phi_4^{(t)} + 2\zeta^{(t)} \Phi_8^{(t)} + 2\delta^{(t)} \Phi_7^{(t)} + 2\eta^{(t)} \Phi_6^{(t)}
\end{equation}
where $\zeta^{(t)}, \delta^{(t)}, \eta^{(t)}$ are the number of pairs of butterflies in $G_s$ which share 0, 1, and 2 edges, respectively. And 
$\Phi_4^{(t)} = \frac{1}{\mu_4^{(t)}} - 1,
\Phi_8^{(t)} = \frac{\mu_8^{(t)}}{(\mu_4^{(t)})^2} - 1,
\Phi_7^{(t)} = \frac{\mu_7^{(t)}}{(\mu_4^{(t)})^2} - 1,
\Phi_6^{(t)} = \frac{\mu_6^{(t)}}{(\mu_4^{(t)})^2} - 1$, where $\mu_j^{(t)}$ is defined as:
$
        \mu_j^{(t)} = \frac{M(M-1)\cdots(M-j+1)}{m_d^{(t)}(m_d^{(t)}-1)\cdots(m_d^{(t)}-j+1)}.
$


\stitle{Complexity Analysis.} We analyze the time and space complexity of Algorithm~\ref{algo:fable} as follows.
Note that \emph{{\cite{sun2024fable} does not provide a similar time and space  complexity analysis}} in their paper.

\begin{theorem}
     Algorithm~\ref{algo:fable} takes  $O(m^{(t)} + (M+M\cdot \ln{\frac{m_d^{(t)}+1}{M}}) \cdot (\log M + d_{{max}}^2))$  time to process $t$ elements in the input bipartite graph stream, where $M$ is the maximum number of edges in the sample, $d_{{max}}$ is the maximum degree of any vertex in the sampled subgraph and $m^{(t)}$/$m_d^{(t)}$ is the number of all/distinct edges up to time~$t$.
\end{theorem}

\begin{proof}
For each new edge, calculating its priority takes \( O(1) \) time. Processing all $m^{(t)}$ edges requires a total time complexity of $O(m^{(t)})$. Each time an edge is added or removed from \( G_s \), the algorithm also performs an insertion or deletion in priority queue $\mathbb{Q}$, both of which take \( O(\log M) \). If an edge is added to \( G_s \), we need to update the butterfly count. This involves checking all pairs of neighbors of \( u \) and \( v \). In the worst case, this complexity is \( O(d_{{max}}^2) \), where \( d_{{max}} \) is the maximum degree of any vertex in the sampled subgraph.
As for the number of edge updates, when $|G_s| < M$, each edge will be sampled. When $G_s$ has reached $M$, due to $h_{\text{max}} \sim \text{Beta}(M, m_d^{(t)} - M + 1)$, each distinct edge is sampled with probability \(\frac{M}{m_d^{(t)}+1} \). Therefore, the total number of sampling operations is \( \sum_{i=M}^{m_d^{(t)}}\frac{M}{i+1} \approx  M\cdot \ln{\frac{m_d^{(t)}+1}{M}}\) based on the approximation formula for harmonic numbers. Therefore, the overall time complexity of Algorithm~\ref{algo:fable} is $O(m^{(t)} + (M+M\cdot \ln{\frac{m_d^{(t)}+1}{M}}) \cdot (\log M + d_{{max}}^2))$.
\end{proof}


\begin{theorem}
    Algorithm~\ref{algo:fable} has a space complexity of $O(M)$, where $M$ is the maximum number of edges in the sampled subgraph.
\end{theorem}

\begin{proof}
    In Algorithm~\ref{algo:fable}, \fable maintains a subgraph $G_s$  that consists of up to $M$ edges. Additionally, the priority queue $\mathbb{Q}$  stores the priority and the corresponding edge for each of the sampled edges, also limited to $M$ entries. Therefore, the space complexity of Algorithm~\ref{algo:fable} is $O(M)$.
\end{proof}


\stitle{Drawbacks of \fable.}  
\revise{
Although the \kmv sketch used in \fable provides an unbiased estimate of the number of distinct edges, $\hat{m}_d$, its variance can be non-negligible, especially when the sample size $M$ is small. Specifically, $\operatorname{Var}(\hat{m}_d) = \frac{(m_d^{(t)})^2}{M - 2}$, introducing additional randomness not accounted for in \fable's variance analysis. In fact, Sun et al.\ implicitly treats $\hat{m}_d$ as deterministic when updating the butterfly count, thereby potentially underestimating the true estimation error. This omission is particularly problematic in low-memory scenarios, where fluctuations of $\hat{m}_d$ can dominate the overall uncertainty in butterfly estimation.
Furthermore, the \fable method requires a priority queue to manage edge priorities. This design not only increases memory consumption but also imposes additional computational overhead due to the maintenance of the priority queue. These drawbacks make \fable less efficient when processing large-scale graphs or when low-latency performance is required, highlighting the need for alternative approaches that have better accuracy and efficiency.
Our proposed \deabc algorithm adopts the bucket-based Flajolet–Martin (\fm) estimator, which naturally aligns with the reservoir design and reduces the variance to $\operatorname{Var}(\hat{m}_d) = \frac{(m_d^{(t)})^2}{1.4426\, M}$, which is approximately 31\% lower than that of \kmv. Moreover, \deabc eliminates the need for a priority queue, thereby reducing both memory consumption and computational overhead, as detailed in the following section.
}

\section{Our Approach}
\label{sec:deabc+}

\revise{In this section, we introduce our \deabcpro algorithm, which does not require any additional storage structures to maintain the priority queue.}
We first present our algorithm, and then analyze its accuracy and complexity. Finally, we provide a remark to discuss the difference between \fable and \deabcpro.


\subsection{\deabcpro Algorithm}



\revise{The main idea of \deabcpro is to maintain a fixed-size sample of edges using a bucket-based priority sampling strategy, where each “bucket” retains only the edge with the lowest priority. In practice, the bucket serves as an abstract concept, while edges are actually stored using adjacency lists to ensure efficient butterfly counting in terms of both storage and computation.
When a new edge arrives, it is mapped to a bucket via a hash function. If the corresponding bucket is empty or the incoming edge has a smaller priority than the currently stored edge, the bucket is updated (implemented by updating the adjacency list entry), and the edge with higher priority is discarded.
To adjust the butterfly count estimates, \deabcpro computes an update factor based on the estimated number of distinct edges and the number of non-empty buckets, where the number of distinct edges is estimated by averaging the priority values across all buckets. 
Compared to \fable, \deabcpro eliminates the need for additional memory to store edges and their associated priorities in the priority queue, thereby improving memory efficiency.
}

\stitle{Algorithm.} The pseudo-code of \deabcpro is presented in Algorithm~\ref{algo:bc2}. We employ $M$ buckets $B[1], \ldots, B[M]$ to store sampled edges. For any new edge $(u, v)$, it is assigned to a bucket based on the hash function $\mathbb{H}{pos}(u, v)$. If the bucket is empty or the new edge has a smaller priority (determined by $\mathbb{H}{pri}(u, v)$) than the edge currently in the bucket, the bucket is updated with the new edge; otherwise, the edge is discarded (lines 3-5). After an edge update occurs, we need to update the estimated number of distinct edges $\hat{m}_d^{(t)}$ (lines 6-13) and update the butterfly count (lines 14-19). The update factor is defined as follows:
\begin{align*}
\theta =
\begin{cases}
\prod_{i=0}^{3} \dfrac{\hat{m}_d^{(t)} - i}{|B_{\neq \emptyset}| - i}, & \hat{m}_d^{(t)} > 3 \text{ and } |B_{\neq \emptyset}| > 3, \\
1, & \text{otherwise}.
\end{cases}
\end{align*}
where $|B_{\neq \emptyset}|$ represents the number of non-empty buckets and $\hat{m}_d^{(t)}$ is estimated number of distinct edges.

To estimate the number of distinct edges in a streaming bipartite graph with duplicate edges, we employ a method inspired by the HyperLogLog algorithm~\cite{flajolet2007hyperloglog}, which uses the Flajolet-Martin (\fm) sketch~\cite{flajolet1985probabilistic} for cardinality estimation. The \fm sketch consists of an integer-valued vector $b = (b[1], b[2], \ldots, b[M])$ of the length of $M$, where each element only stores the maximum integer value.
The method involves hashing each edge an integer value that follows $Geometric(1/2)$ distribution, and randomly assigns it to a bucket. Let $P(b^{(t-1)})$ is the probability of any bucket $b[i]^{(t-1)}, i \in {1,\cdots, M}$ being updated when a new value $y^{t}$ arrives, which is defined as:
\begin{align*}
P(b^{(t-1)}) &= \frac{1}{M}\sum_{j=1}^{M}P\left( y^{t} > b[j]^{(t-1)}\right)
=\frac{1}{M}\sum_{j=1}^{M}2^{-b[j]^{(t-1)}}.
\end{align*}

Accounting to Ting \cite{DBLP:conf/kdd/Ting14}, the estimated number of distinct edges $\hat{m}_d^{(t)}$ at the end of time $t$ can be expressed as:
\begin{align}
\label{n_edge_t}
\hat{m}_d^{(t)} = \hat{m}_d^{(t-1)} + \frac{\mathds{1}\left(b[i]^{(t)} \neq b[i]^{(t-1)}\right)}{P(b^{(t-1)})} \text{ for any } i \in {1, \cdots, M},
\end{align}
where $\mathds{1}(P) $ an indicator function that returns 1 if the condition $P$ is true, and 0 otherwise.

In our method, we represent $Geometric(1/2)$ distribution using $\rho = -\lfloor\log \mathbb{H}_{pri}(B[pos])\rfloor$
Due to $\mathbb{H}_{pri} \sim Uniform(0, 1)$, $P(\rho = k) = P(k \le -\lfloor\log \mathbb{H}_{pri}(u,v)\rfloor < k+1) =  P(2^{-(k+1)} \le \mathbb{H}_{pri}(u,v) < 2^{-k}) = 2^{-k} - 2^{-(k+1)} = 2^{-(k+1)}$, which equals to $Geometric(1/2)$.
We use a counter $q = \frac{1}{M}\sum_{i=1}^M2^{-\rho}$ to update the value of $\hat{m}_d^{(t)}$.
When the value of $q$ is updated, that is, when the $\rho$ of a new edge is greater than the $\rho_{max}$ of the old edge already assigned to the bucket, we update the value of $\hat{m}_d^{(t)}$ as $\hat{m}_d^{(t)} \gets \hat{m}_d^{(t-1)} + \frac{1}{q}$.

\begin{figure}[t]\centering
    \scalebox{0.35}[0.35]{\includegraphics{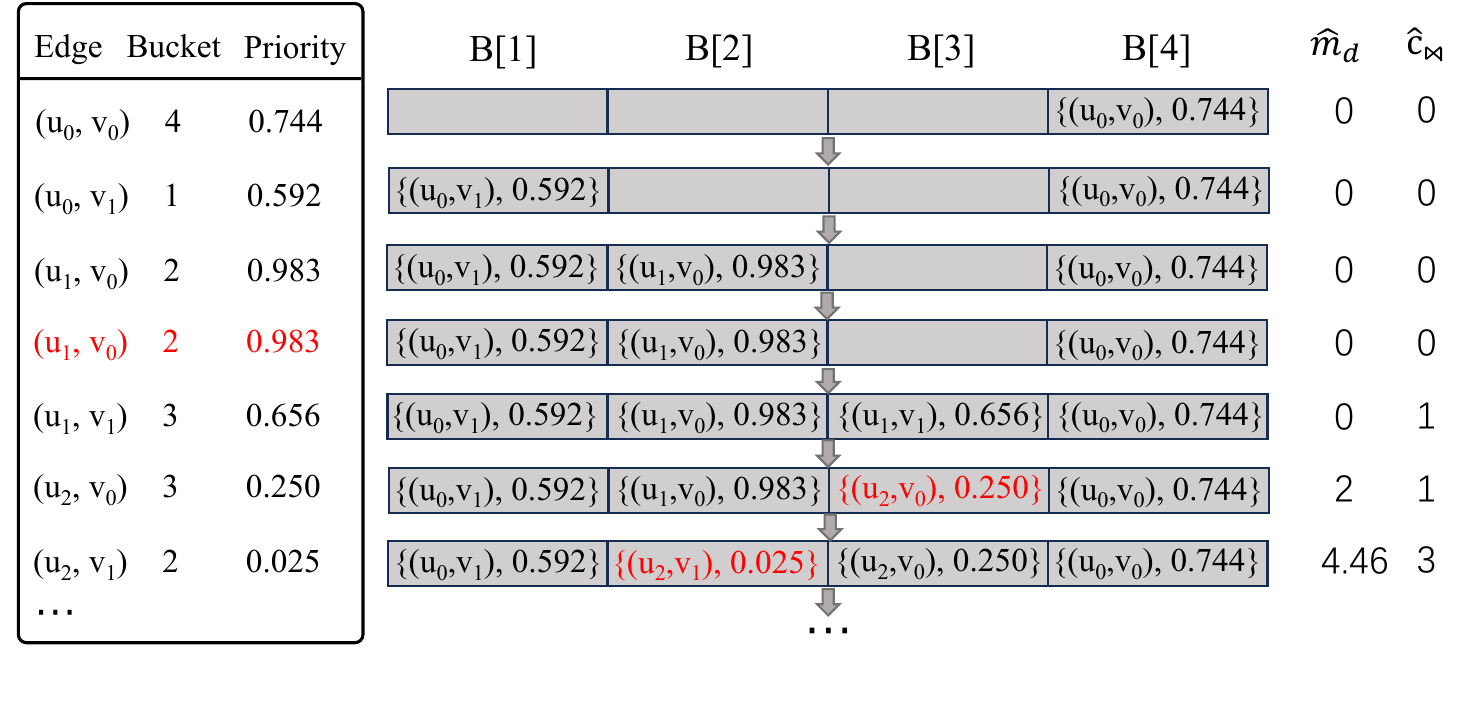}}
     \vspace{-1em}
    \caption{\revise{An Example of \deabcpro Algorithm }}
     \vspace{-1.5em} 
    \label{fig:al}
\end{figure}

\stitle{\revise{\underline{Example:}}}
\revise{Consider a bipartite edge stream depicted in Figure~\ref{fig:al}. The \deabcpro algorithm uses four buckets ($B[1]$ to $B[4]$) for sampling edges. Each incoming edge is assigned to a bucket using a hash function and is associated with a priority value.
When a duplicate edge (e.g., the second $(u_1,v_0)$ in step 4) arrives, it hashes to the same bucket with the same priority. Since the bucket already contains this edge, no replacement or update occurs. 
Each bucket update potentially increments the distinct edge count estimate $\hat{m}_d$.
In Step 7, the edge $(u_2,v_1)$, assigned to bucket $B[2]$, replaces the previous edge $(u_1,v_0)$ due to its lower priority. After this replacement, the estimated number of distinct edges is updated to 4.46 (lines 6–13 in Algorithm~\ref{algo:bc2}). The insertion of edge $(u_2,v_1)$ forms the new butterfly $\{u_0, v_1, u_2, v_0\}$, and based on the update factor $\theta = \frac{4.46(4.46-1)(4.46-2)(4.46-3)}{4(4-1)(4-2)(4-3)} \approx 2.309$, the butterfly count is updated to 3.
}



 
\begin{algorithm}[h]

    \DontPrintSemicolon
    \caption{DEABC}
    \label{algo:bc2}
    \KwIn{graph stream $\Pi$, maximum edge number $M$, and hash functions $\mathbb{H}_{pri}$ and $\mathbb{H}_{pos}$}
    \KwOut{the number of butterflies $\hat{c}_{\Join}$}
    \SetKwComment{comment}{$\triangleright$ }{}

    $\hat{m}_d^{(t)} \gets 0; q = 1; \hat{c}_{\Join} \gets 0$;
    
    Initialize all buckets $B[1,\ldots,M] \gets [\emptyset,\ldots,\emptyset]$ and $|B_{\neq \emptyset}|$ represents the number of non-empty buckets;

    \For{\textbf{each} edge $(u, v) \in \Pi$}{
         $pos \gets \mathbb{H}_{pos}(u, v)$; \\
         \If{$B[pos] = \emptyset$ or $\mathbb{H}_{pri}(u,v) < \mathbb{H}_{pri}(B[pos])$}{
             \If{$B[pos] \neq \emptyset$}{remove the edge in $B[pos]$; $\rho_{max} \gets -\lfloor\log \mathbb{H}_{pri}(B[pos])\rfloor$;}
             \Else{ $\rho_{max} \gets 0$;}

             $\rho \gets -\lfloor\log \mathbb{H}_{pri}(u,v)\rfloor$;

             \If{$\rho > \rho_{max}$}{
                 $\hat{m}_d^{(t)} \gets \hat{m}_d^{(t-1)} + \frac{1}{q}$;\\
                 $q \gets q + \frac{1}{M}(2^{-\rho} - 2^{-\rho_{max}})$;\\
             }

             \If{$\hat{m}_d^{(t)} > 3$ and $|B_{\neq \emptyset}| > 3$}{
                 $\theta \gets \prod_{i=0}^{3} \frac{\hat{m}_d^{(t)} - i}{|B_{\neq \emptyset}| - i}$;
             }

             \Else{$\theta \gets 1$;}

             $\hat{c}_{\Join} \gets \hat{c}_{\Join} +$ \textbf{\FButterflies($(u,v), \theta$)};\\

             $B[pos] \gets (u,v)$;\\
         }
    }


\end{algorithm}

\subsection{Accuracy Analysis}

We analyze the accuracy of \deabcpro, demonstrating that it provides unbiased estimates of the butterfly count with bounded variance.

\subsubsection{Unbiasedness} Using the following theorem, we establish the unbiasedness of the \deabcpro algorithm.

\begin{theorem}
\label{the:unbiasedness2}
The \deabcpro algorithm provides an unbiased estimate of the butterfly count. Specifically, 
$
\mathbb{E}[\hat{c}_{\Join}] = c_{\Join}.
$
where \( \hat{c}_{\Join} \) is the butterfly count estimate produced by \deabcpro at any time \( t \) and \( c_{\Join} \) is the true butterfly count.
\end{theorem}

\begin{proof}

To prove that the estimate of the number of butterflies \( \hat{c_{\Join}} \) produced by our \deabc is unbiased, we need to show that \( \mathbb{E}[\hat{c}_{\Join}] = c_{\Join} \), where \( c_{\Join} \) is the true number of butterflies in the graph. We define a random variable \( X_{\Join} \), which represents the number of times the butterfly \( \Join \) is counted, as follows:
\begin{align}
\label{butterfly_random_variable}
X_{\Join}=
\begin{cases}
\prod_{i=0}^{3} \frac{\hat{m}_d^{(t)} - i}{|B_{\neq \emptyset}| - i} & \Join \text{is counted}, \\
0 & \text{otherwise}.
\end{cases}
\end{align}

Based on Equation \ref{n_edge_t}, we can obtain that the expectation of $\hat{m}_d^{(t)}$ at time $t$ is, 

\begin{align*}
\mathbb{E}\left[\hat{m}_d^{(t)}\right] &= \mathbb{E}\left[\hat{m}_d^{(t-1)}\right]
+ \mathbb{E}\left[\frac{\mathds{1}\left(b[i]^{(t)} \neq b[i]^{(t-1)}\right)}{P(b^{(t-1)})}\right] \\
&=\mathbb{E}\left[\hat{m}_d^{(t-1)}\right] 
+\frac{\frac{1}{M}\sum_{j=1}^{M}P\left( b[j]^{(t)} > b[j]^{(t-1)}\right)}{P(b^{(t-1)})} \\
&=\mathbb{E}\left[\hat{m}_d^{(t-1)}\right] + 1.
\end{align*}

When $t=1$, $\mathbb{E}\left[\hat{m}_d^{(1-1)}\right] = m_d^{(0)} = 0$.
When a new distinct edge arrives, the expectation of the number of distinct edge count increases by 1, therefore, we can ultimately derive recursively that $\mathbb{E}\left[\hat{m}_d^{(t)}\right] = m_d^{(t)}$ for any $t$, implying that the estimation of the number of distinct edges is unbiased. 


According to~\cite{DBLP:journals/pvldb/WangQSZTG17}, when a butterfly $\Join$ is counted at the end of time $t$, 
the probability that all four edges forming $\Join$ are sampled is as follows:
\[
P(\Join) = \frac{|B_{\neq \emptyset}|(|B_{\neq \emptyset}|-1)(|B_{\neq \emptyset}|-2)(|B_{\neq \emptyset}|-3)}{{m_d^{(t)}}({m_d^{(t)}}-1)({m_d^{(t)}}-2)({m_d^{(t)}}-3)}.
\]
where $|B_{\neq \emptyset}|$ is the number of non-empty buckets, and ${m_d^{(t)}}$ is the number of distinct edges.

And then,
\begin{align*}
\mathbb{E}[X_{\Join}\mid \hat{m}_d^{(t)}] &= X_{\Join}\cdot P(\Join) = \prod_{i=0}^{3} \frac{\hat{m}_d^{(t)} - i}{m_d^{(t)} - i} .
\end{align*}


Let $f_1(\hat{m}_d^{(t)}) = \mathbb{E}[X_{\Join}\mid \hat{m}_d^{(t)}]$, and perform a first-order Taylor expansion at $m_d^{(t)}$:
\[
f_1(\hat{m}_d^{(t)}) = f_1(m_d^{(t)}) + f_1'(m_d^{(t)})(\hat{m}_d^{(t)} - m_d^{(t)}) + O(h).
\]

Therefore, 
\begin{align*}
\mathbb{E}[X_{\Join}] &= \mathbb{E}[f_1(\hat{m}_d^{(t)})] = f_1(m_d^{(t)}) 
= \prod_{i=0}^{3} \frac{m_d^{(t)} - i}{m_d^{(t)} - i} = 1.
\end{align*}

$\mathbb{E}[\hat{c}_{\Join}] = \sum_{\Join \in G} \mathbb{E}[X_{\Join}] = |c_{\Join}| \times 1 = c_{\Join}$. Therefore, Our algorithm’s estimate of the number of butterflies in Algorithm \ref{algo:bc2} is unbiased.

\end{proof}

\subsubsection{Variance} The variance of the \deabcpro algorithm is stated in the following theorem.
\begin{theorem}
\label{the:variance2}
The variance of the \deabcpro algorithm's butterfly count estimate \( \hat{c}_{\Join} \) is as follows:
\begin{align*}
\mathrm{Var}(\hat{c}_{\Join})
&= c_{\Join}  \Phi_1'  + 2\zeta^{(t)}  \Phi_2' + 2\delta^{(t)}  \Phi_3' + 2\eta^{(t)} \Phi_4' - c_{\Join}^2 \\
&+ (g'(m_d^{(t)}))^2\cdot \frac{(m_d^{(t)})^2}{1.4426M}.
\end{align*}
where $\zeta^{(t)}, \delta^{(t)}, \eta^{(t)}$ are the number of pairs of butterflies in $G_s$ which share 0, 1, and 2 edges up to time $t$, respectively. $g(\hat{m}_d^{(t)})=\frac{\hat{m}_d^{(t)}(\hat{m}_d^{(t)}-1)(\hat{m}_d^{(t)}-2)(\hat{m}_d^{(t)}-3)}{m_d^{(t)}(m_d^{(t)}-1)(m_d^{(t)}-2)(m_d^{(t)}-3)}\cdot c_{\Join}$. And
\[
\Phi_1' = \frac{m_d^{(t)}(m_d^{(t)}-1)(m_d^{(t)}-2)(m_d^{(t)}-3)}{|B_{\neq \emptyset}|(|B_{\neq \emptyset}|-1)(|B_{\neq \emptyset}|-2)(|B_{\neq \emptyset}|-3)},
\]
\[
\Phi_2' = \Phi_1' \cdot \frac{(|B_{\neq \emptyset}| -4) (|B_{\neq \emptyset}|-5)  (|B_{\neq \emptyset}|-6) (|B_{\neq \emptyset}|-7)}{(m_d^{(t)} - 4)  (m_d^{(t)}-5)  (m_d^{(t)}-6)  (m_d^{(t)}-7)},
\]
\[
\Phi_3' = \Phi_1' \cdot \frac{(|B_{\neq \emptyset}| -4) (|B_{\neq \emptyset}|-5)  (|B_{\neq \emptyset}|-6)}{(m_d^{(t)} - 4)  (m_d^{(t)}-5)  (m_d^{(t)}-6) },
\]
\[
\Phi_4' = \Phi_1' \cdot \frac{(|B_{\neq \emptyset}| -4) (|B_{\neq \emptyset}|-5)}{(m_d^{(t)} - 4)  (m_d^{(t)}-5)}.
\]

\end{theorem}

\begin{proof}

Considering the bias of $\hat{m}_d^{(t)}$, we can calculate the variance by applying the Law of Total Variance:
\[
\mathrm{Var}(\hat{c}_{\Join}) = \mathbb{E}[\mathrm{Var}(\hat{c}_{\Join}|\hat{m}_d^{(t)})] + \mathrm{Var}(\mathbb{E}[\hat{c}_{\Join}|\hat{m}_d^{(t)}]).
\]

First, we calculate $\mathbb{E}[\mathrm{Var}(\hat{c}_{\Join}|\hat{m}_d^{(t)})]$. According to the definition of variance, we can get:
\begin{align*}
&\mathrm{Var}(\hat{c}_{\Join}\mid \hat{m}_d^{(t)}) = \mathbb{E}[\hat{c}_{\Join}^2\mid \hat{m}_d^{(t)}] - (\mathbb{E}[\hat{c}_{\Join}\mid \hat{m}_d^{(t)}])^2 \\
&= \sum_i \mathbb{E}[X_i^2\mid \hat{m}_d^{(t)}] +\sum_{i \neq j}\mathbb{E}[X_i X_j\mid \hat{m}_d^{(t)}] - (\mathbb{E}[\hat{c}_{\Join}\mid \hat{m}_d^{(t)}])^2 .
\end{align*}

According to the definition of $X_{\Join}$ in Equation~\ref{butterfly_random_variable}, we can calculate:
\begin{align*}
&\mathbb{E}[X_i^2\mid \hat{m}_d^{(t)}] = \left(\prod_{i=0}^{3} \frac{\hat{m}_d^{(t)} - i}{|B_{\neq \emptyset}| - i}\right)^2 \cdot P(\Join) \\
&=\frac{\prod_{i=0}^{3}(\hat{m}_d^{(t)} - i)^2}{\prod_{i=0}^{3}(|B_{\neq \emptyset}| - i)^2} \cdot \frac{|B_{\neq \emptyset}|(|B_{\neq \emptyset}|-1)\cdots(|B_{\neq \emptyset}|-3)}{{m_d^{(t)}}({m_d^{(t)}}-1)\cdots({m_d^{(t)}}-3)}.
\end{align*}

When calculating $\mathbb{E}[X_i X_j]$, we distinguish the following three cases: $(i)$ the butterflies $i$ and $j$ do not share any edge; $(ii)$ the butterflies $i$ and $j$ only share one edge; and $(iii)$ the butterflies $i$ and $j$ share two edges. Note that it is impossible for two distinct butterflies to share three or more edges, because they would be the same butterfly. The probability that both butterflies $i$ and $j$ are counted is,

\begin{align}
P_{c3} = 
\begin{cases}
\frac{|B_{\neq \emptyset}|  (|B_{\neq \emptyset}|-1)  (|B_{\neq \emptyset}|-2)  \cdots  (|B_{\neq \emptyset}|-7)}{m_d^{(t)}  (m_d^{(t)}-1)  (m_d^{(t)}-2)  \cdots (m_d^{(t)}-7)} & case (i), \\
\frac{|B_{\neq \emptyset}|  (|B_{\neq \emptyset}|-1)  (|B_{\neq \emptyset}|-2)  \cdots  (|B_{\neq \emptyset}|-6)}{m_d^{(t)}  (m_d^{(t)}-1)  (m_d^{(t)}-2)  \cdots (m_d^{(t)}-6)} & case (ii), \\
\frac{|B_{\neq \emptyset}|  (|B_{\neq \emptyset}|-1)  (|B_{\neq \emptyset}|-2) \cdots  (|B_{\neq \emptyset}|-5)}{m_d^{(t)}  (m_d^{(t)}-1)  (m_d^{(t)}-2)  \cdots (m_d^{(t)}-5)} & case(iii), \\
0 & \text{otherwise}.
\end{cases}
\end{align}

Then,

\begin{align*}
&\mathbb{E}[X_iX_j\mid \hat{m}_d^{(t)}] = \left(X_{\Join}\right)^2 \cdot \left(p_{1}P_{c3}^{(i)} + p_{2}P_{c3}^{(ii)} + p_{3}P_{c3}^{(iii)}\right) .
\end{align*}
where $p_1, p_2, p_3$ represent the probabilities of case $i, ii, iii$ and $p_1+p_2+p_3 = 1$.

Therefore,
$\mathrm{Var}(\hat{c}_{\Join}\mid \hat{m}_d^{(t)}) = \sum_i \mathbb{E}[X_i^2\mid \hat{m}_d^{(t)}] +\sum_{i \neq j}\mathbb{E}[X_i X_j\mid \hat{m}_d^{(t)}] - (\mathbb{E}[\hat{c}_{\Join}\mid \hat{m}_d^{(t)}])^2$.
Let $f_2(\hat{m}_d^{(t)}) = \mathrm{Var}(\hat{c}_{\Join}\mid \hat{m}_d^{(t)})$, we perform a first-order Taylor expansion in $m_d^{(t)}$, 
\[
f_2(\hat{m}_d^{(t)}) = f_2(m_d^{(t)}) + f_2'(m_d^{(t)})(\hat{m}_d^{(t)} - m_d^{(t)}) + O(h).
\]

Then we can compute,
\begin{align*}
&\mathbb{E}[\mathrm{Var}(\hat{c}_{\Join}|\hat{m}_d^{(t)})] = \mathbb{E}[f_2(\hat{m}_d^{(t)})] = f_2(m_d^{(t)}) \\
&=\sum_i \mathbb{E}[X_i^2\mid  m_d^{(t)}] +\sum_{i \neq j}\mathbb{E}[X_i X_j\mid  m_d^{(t)}] - c_{\Join}^2\\
&=c_{\Join}  \Phi_1'  + 2\zeta^{(t)}  \Phi_2' + 2\delta^{(t)}  \Phi_3' + 2\eta^{(t)} \Phi_4' - c_{\Join}^2.
\end{align*}
where $\zeta^{(t)}, \delta^{(t)}, \eta^{(t)}$ are the number of pairs of butterflies in $G_s$ which share 0, 1, and 2 edges up to time $t$, respectively. And $\Phi_1', \Phi_2', \Phi_3', \Phi_4'$ are shown in Theorem~\ref{the:variance2}.



Then we calculate $\mathrm{Var}(\mathbb{E}[\hat{c}_{\Join}|\hat{m}_d^{(t)}])$. Let $g(\hat{m}_d^{(t)}) = \mathbb{E}[\hat{c}_{\Join}|\hat{m}_d^{(t)}] = \sum_{\Join \in G} (X_{\Join}\cdot P(\Join))=\frac{\hat{m}_d^{(t)}(\hat{m}_d^{(t)}-1)(\hat{m}_d^{(t)}-2)(\hat{m}_d^{(t)}-3)}{m_d^{(t)}(m_d^{(t)}-1)(m_d^{(t)}-2)(m_d^{(t)}-3)}\cdot c_{\Join}$, we further perform a first-order Taylor expansion in $m_d^{(t)}$, 
\[
g(\hat{m}_d^{(t)}) = g(m_d^{(t)}) + g'(m_d^{(t)})(\hat{m}_d^{(t)} - m_d^{(t)}) + O(h).
\]

So $\mathrm{Var}(g(\hat{m}_d^{(t)})) = (g'(m_d^{(t)}))^2\mathrm{Var}(\hat{m}_d^{(t)})$. According to \cite{DBLP:conf/kdd/Ting14}, the variance of $\hat{m}_d^{(t)}$ is $\mathrm{Var}(\hat{m}_d^{(t)}) \approx \frac{(m_d^{(t)})^2}{1.4426M}$.

Overall, we can obtain the variance,
\begin{align*}
\mathrm{Var}(\hat{c}_{\Join}) &= \mathbb{E}[\mathrm{Var}(\hat{c}_{\Join}|\hat{m}_d^{(t)})] + \mathrm{Var}(\mathbb{E}[\hat{c}_{\Join}|\hat{m}_d^{(t)}])\\
&= c_{\Join}  \Phi_1'  + 2\zeta^{(t)}  \Phi_2' + 2\delta^{(t)}  \Phi_3' + 2\eta^{(t)} \Phi_4' - c_{\Join}^2 \\
&+ (g'(m_d^{(t)}))^2\cdot \frac{(m_d^{(t)})^2}{1.4426M}.
\end{align*}

\end{proof}

\subsubsection{Error Analysis}

Using Chebyshev's inequality, we can bound the probability that the estimate deviates from its expected value by more than a certain multiple of the standard deviation.

\begin{corollary}
For any constant \( \lambda > 0 \), the probability that the estimate \( \hat{c}_{\Join} \) deviates from its expected value by more than \( \lambda \) times the standard deviation is bounded by \( \frac{1}{\lambda^2} \). Formally:
\[
P\left(\left|\hat{c}_{\Join} - \mathbb{E}[\hat{c}_{\Join}]\right| \geq \lambda \times \sqrt{\mathrm{Var}(\hat{c}_{\Join})}\right) \leq \frac{1}{\lambda^2}.
\]
\end{corollary}

\begin{proof}
This result follows directly from applying the Chebyshev's inequality.
\end{proof}

\subsection{Complexity Analysis}



\begin{theorem}
     Algorithm~\ref{algo:bc2} takes  $O(m^{(t)} + M\cdot \ln{(1+\frac{m_d^{(t)}}{M}}) \cdot  d_{{max}}^2)$ time to process $t$ elements in the input bipartite graph stream, where $M$ is the maximum number of edges in the sampled subgraph, $d_{{max}}$ is the  maximum degree of any vertex in the sampled subgraph and $m^{(t)}/m_d^{(t)}$ is the number of all/distinct edges up to time $t$.
\end{theorem}

\begin{proof}
Processing all edges in lines 4-5 requires a total $m^{(t)}$ time.
Each time a new edge replaces any edge in buckets, we need to update the butterfly count, which takes $O(d_{max}^2)$ time in the worst case, and the distinct edge count, which takes $O(1)$ time. 
Due to $\mathbb{H}_{pri} \sim Uniform(0, 1)$, the hash value of each edge in buckets follows $Beta(1, m_{B[i]}^{(t)})$ distribution, where $m_{B[i]}^{(t)}$ is the number of distinct edges assigned to bucket $B[i]$ up to time $t$. The total number of edge replacements is $\sum_{i=1}^M\sum_{j=0}^{m_{B[i]}^{(t)}}\frac{1}{j+1}$. Due to $\mathbb{H}_{pos} \sim DiscreteUniform(1,M)$, we can conclude that the number of edges assigned to each bucket is approximately  $\frac{m_d^{(t)}}{M}$ when $m_d^{(t)} \gg M$. Therefore, $\sum_{i=1}^M\sum_{j=0}^{m_{B[i]}^{(t)}}\frac{1}{j+1} = M\cdot\sum_{j=0}^{\frac{m_d^{(t)}}{M}}\frac{1}{j+1} \approx M\cdot \ln{(1+\frac{m_d^{(t)}}{M})}$ based on the approximation formula for harmonic numbers. The overall time complexity of Algorithm~\ref{algo:bc2} is $O(m^{(t)} + M\cdot \ln{(1+\frac{m_d^{(t)}}{M}}) \cdot  d_{{max}}^2)$.
\end{proof}


\begin{theorem}
    Algorithm~\ref{algo:bc2} has a space complexity of $O(M)$, where $M$ is the maximum number of edges in the sampled subgraph.
\end{theorem}

\begin{proof}
In Algorithm~\ref{algo:bc2}, \deabcpro maintains $M$ buckets, where each bucket \( B[1, \dots, M] \) stores at most one edge. Therefore, the space complexity of Algorithm~\ref{algo:bc2} is $O(M)$.
\end{proof}

\stitle{\underline{Remark}} The variance difference between our \deabcpro and \fable primarily stems from the bias in estimating the number of distinct edges $\hat{m}_d$. By accounting for the bias in  $\hat{m}_d$, the variance of \fable can be adjusted to: $\mathrm{Var}(\hat{c}_{\Join}) = c_{\Join} \Phi_4^{(t)} + 2\zeta^{(t)} \Phi_8^{(t)} + 2\delta^{(t)} \Phi_7^{(t)} + 2\eta^{(t)} \Phi_6^{(t)} + (g'(m_d^{(t)}))^2\cdot \frac{(m_d^{(t)})^2}{M-2}$, where $g(\hat{m}_d^{(t)})=\frac{\hat{m}_d^{(t)}(\hat{m}_d^{(t)}-1)(\hat{m}_d^{(t)}-2)(\hat{m}_d^{(t)}-3)}{m_d^{(t)}(m_d^{(t)}-1)(m_d^{(t)}-2)(m_d^{(t)}-3)}\cdot c_{\Join}$.
Compared to the \kmv sketch ($\mathrm{Var}(\hat{m}_d) = \frac{(m_d)^2}{M-2}$) that \fable uses, our \deabcpro achieves lower variance ($\mathrm{Var}(\hat{m}_d) = \frac{(m_d)^2}{1.4426M}$).
Analyzing the time complexity of \deabc and \deabcpro reveals that the logarithmic terms in \fable $(\ln{\frac{m_d^{(t)}+1}{M}})$ and our \deabcpro $(\ln{(1+\frac{m_d^{(t)}}{M^2}}))$ are nearly equivalent due to $ m_d(t) \gg M$. However, \deabcpro exhibits linear dependence on $M$, while \fable’s time complexity grows faster due to the additional $\log M$ factor. As a result, \deabcpro demonstrates superior time efficiency when processing large-scale data, particularly for large $M$.
In terms of space complexity, both \fable and \deabcpro require $O(M)$ space. However, \fable demands additional memory to maintain a priority queue alongside the sampled edges, effectively doubling its memory usage compared to \deabcpro.
In summary, compared to \fable, our \deabcpro offers better time efficiency, reduced memory usage, and lower variance guarantees, enhancing both performance and reliability when processing large-scale data.

Several prior works, such as \mascot~\cite{lim2015mascot}, \triest~\cite{stefani2017triest} and \abacus~\cite{DBLP:conf/icde/PapadiasKPQM24} update the count when an edge arrives rather than when an edge is sampled as in our method. 
\reviseminor{Such a strategy typically leads to more accurate estimations, but it is not directly applicable to bipartite streams with duplicate edges. When an edge is not sampled, it is difficult to reliably determine whether it has already appeared.}
This uncertainty may cause the same butterfly to be counted multiple times, thereby making unbiased estimation challenging in practice.

\section{Experimental Evaluation}
\label{sec:ee}

\begin{table}[h]
    \caption{Datasets}
    \label{tab:syn_data}
    \resizebox{0.48\textwidth}{!}{
    \begin{tabular}{l|c|c|c|c}
    \hline
    Datasets & $|L|$ & $|R|$ & $|E|$ & Num. of Butterflies \\ \hline

    Twitter-ut           &   175,215     &   530,419   &   4,664,605    &  206,474,165     \\ \hline
    Edit-frwiki           &   94,307     &   62,690   &   5,733,289    &   34,161,632,174     \\ \hline
    AmazonRatings           &   2,146,058     &   1,230,916   &   5,838,041    &  35,679,455     \\ \hline
    Movie-lens            & 69,878       &  10,677   &  10,000,054    &   1,197,015,057,161      \\ \hline
    Edit-itwiki            &  95,210      &   623,245   &  13,210,412   &  171,351,772,803   \\ \hline
    Lastfm-band           &   993     &   174,078   &   19,150,868    &  4,419,565,780     \\ \hline
    Discogs           &   1,617,944     &   384   &   24,085,580    &  77,383,418,076     \\ \hline
    Yahoo-songs            &  130,560     &   136,738   &   49,770,695   &   2,323,845,186,838     \\ \hline
    StackoverFlow        &  6,024,261      &   6,024,271   &  63,497,050   &  26,740,122,639    \\ \hline
    LiveJournal            &  3,201,203      &  7,489,073    &   112,307,385   &  2,483,262,136,916    \\ \hline

    Delicious-ui            &  833,082     &   33,778,222   &  301,186,579   &  56,892,252,403    \\ \hline
    Orkut            &   2,783,196     &   8,730,857   &   327,037,487   &   21,996,969,752,862   \\ \hline

    \end{tabular}
    }
\end{table}

\begin{figure*}[]
    \centering
    \vspace{-0.8em}
    \includegraphics[width=0.52\textwidth]{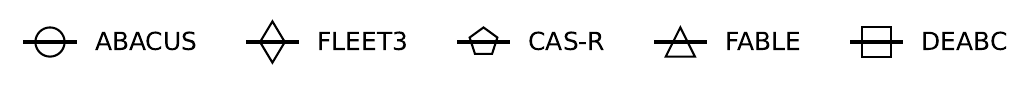}

    \begin{minipage}{0.22\textwidth}
        \centering
        \includegraphics[width=\textwidth]{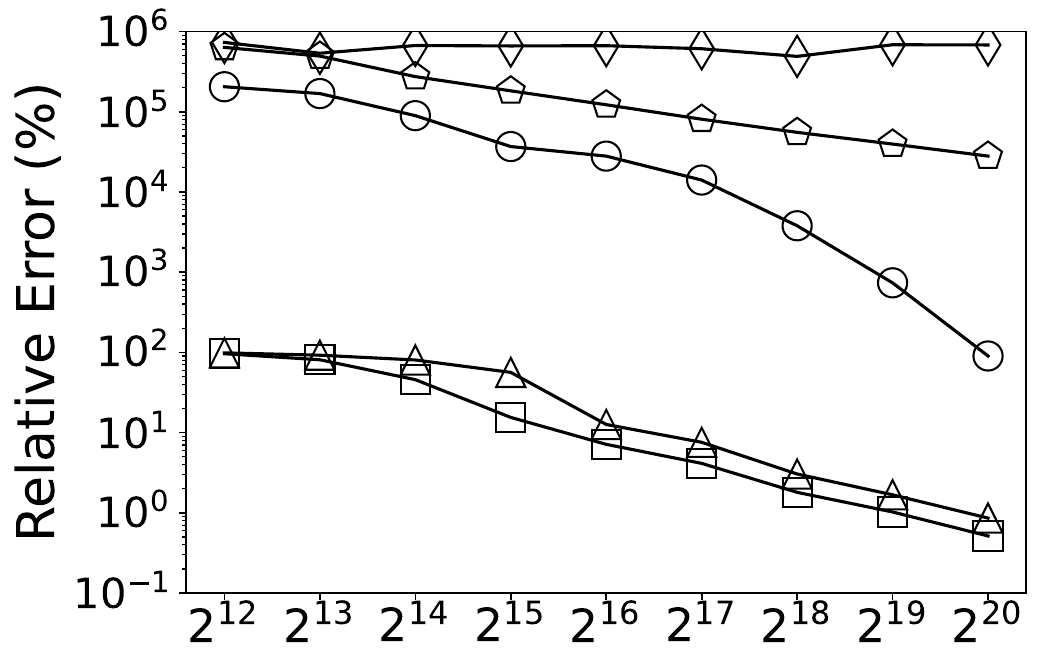}
        \subcaption{Twitter-ut (edges)}
    \end{minipage}
    \hspace{0.015\linewidth}
    \begin{minipage}{0.22\textwidth}
        \centering
        \includegraphics[width=\textwidth]{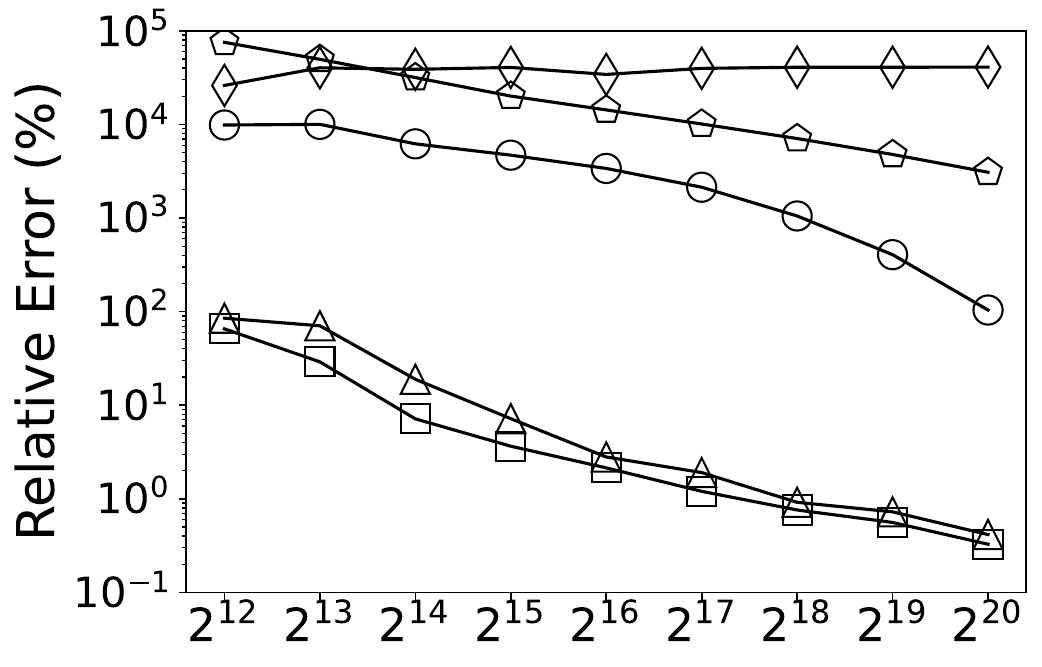}
        \subcaption{Edit-frwiki (edges)}
    \end{minipage}
    \hspace{0.015\linewidth}
    \begin{minipage}{0.22\textwidth}
        \centering
        \includegraphics[width=\textwidth]{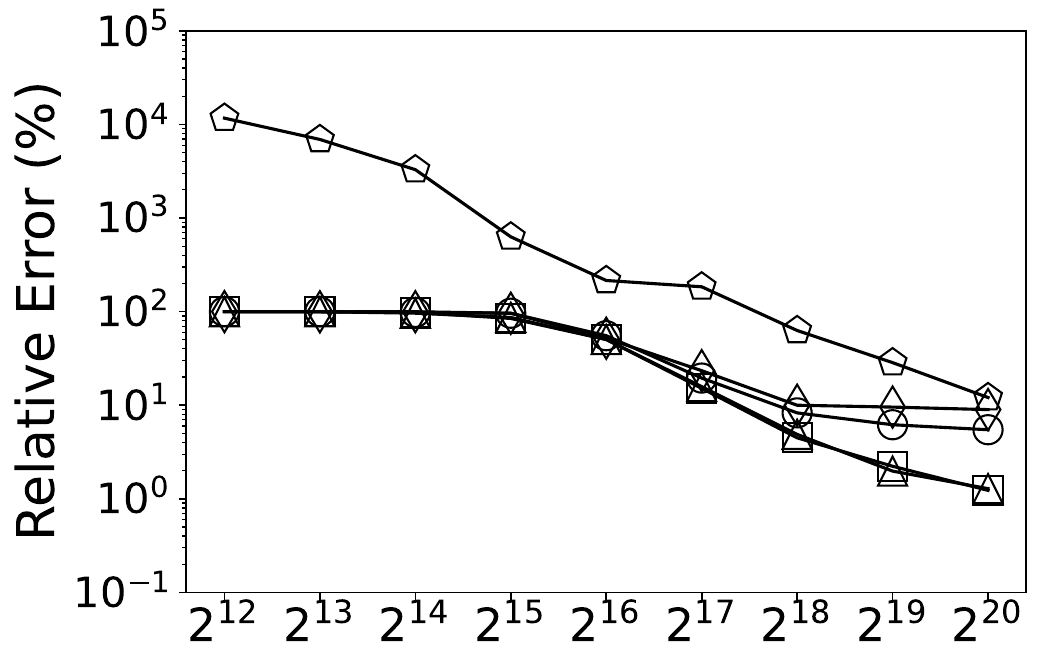}
        \subcaption{AmazonRatings (edges)}
    \end{minipage}
    \hspace{0.015\linewidth}
    \begin{minipage}{0.22\textwidth}
        \centering
        \includegraphics[width=\textwidth]{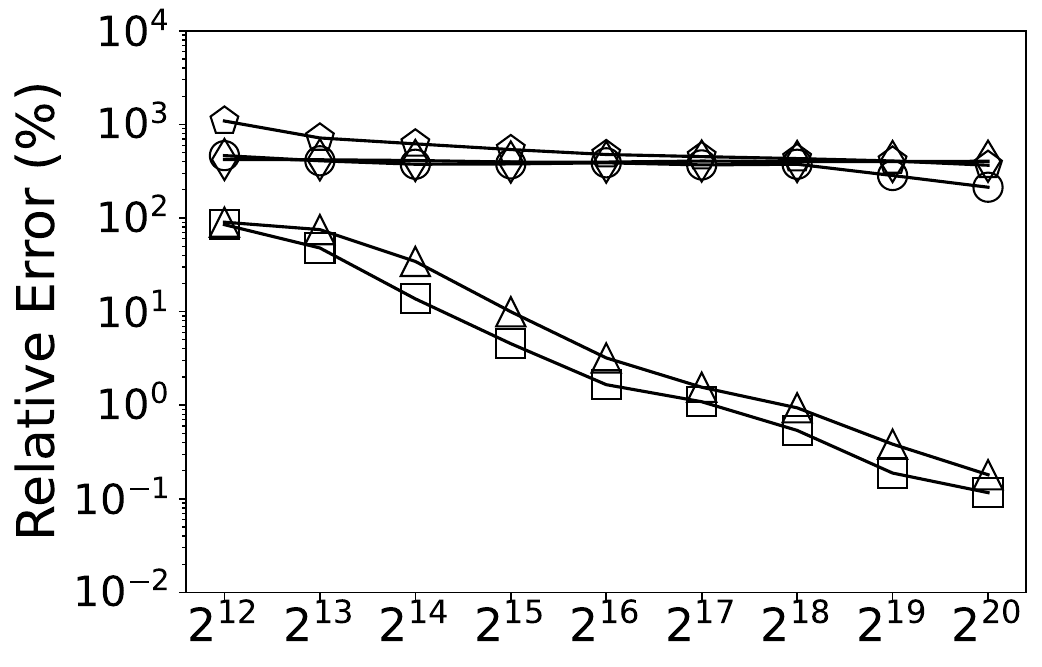}
        \subcaption{Movie-lens (edges)}
    \end{minipage}
    
    \begin{minipage}{0.22\textwidth}
        \centering
        \includegraphics[width=\textwidth]{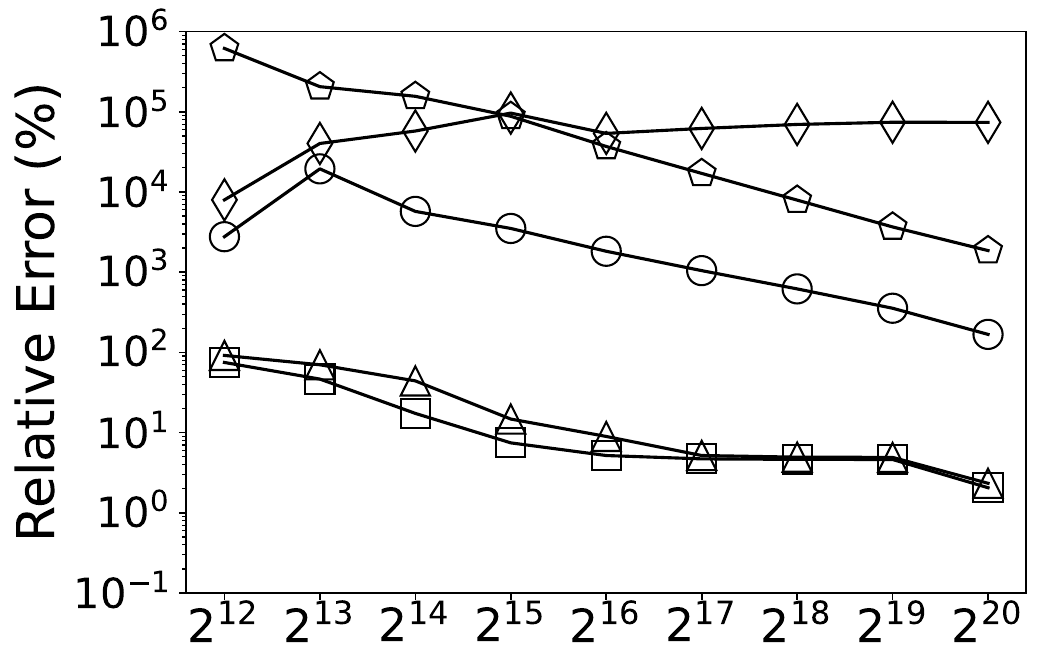}
        \subcaption{Edit-itwiki (edges)}
    \end{minipage}
    \hspace{0.015\linewidth}
    \begin{minipage}{0.22\textwidth}
        \centering
        \includegraphics[width=\textwidth]{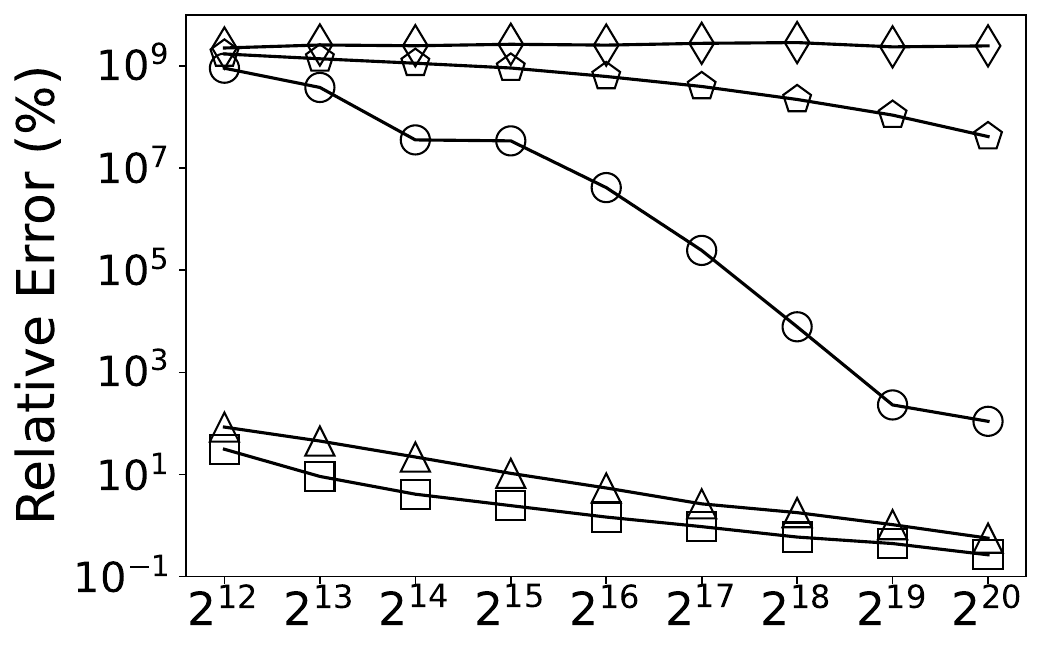}
        \subcaption{Lastfm-band (edges)}
    \end{minipage}
    \hspace{0.015\linewidth}
    \begin{minipage}{0.22\textwidth}
        \centering
        \includegraphics[width=\textwidth]{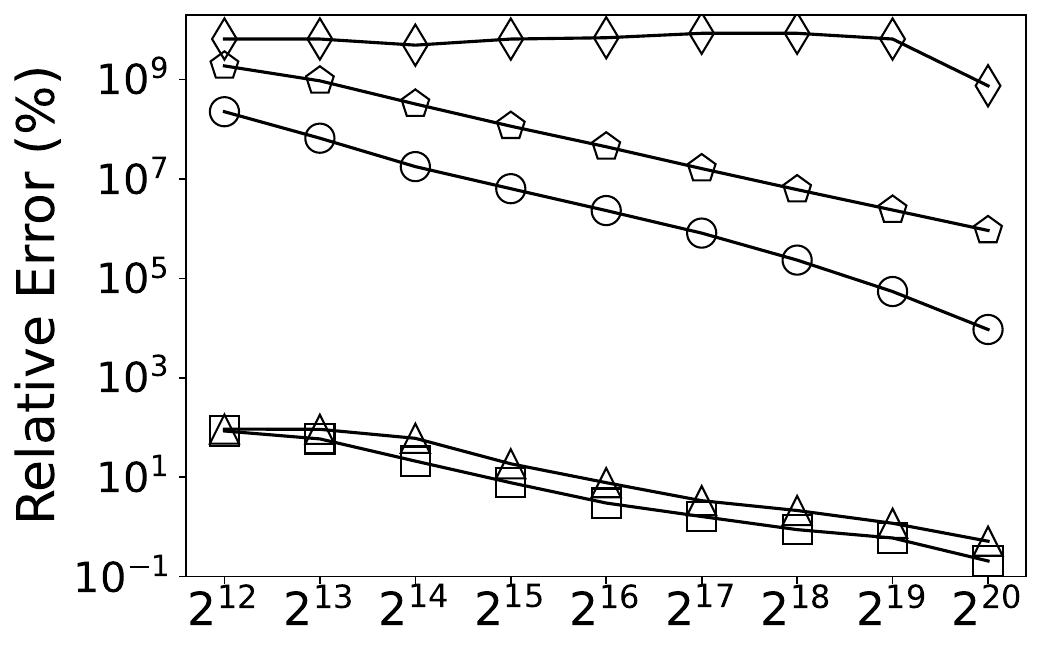}
        \subcaption{Discogs (edges)}
    \end{minipage}
    \hspace{0.015\linewidth}
    \begin{minipage}{0.22\textwidth}
        \centering
        \includegraphics[width=\textwidth]{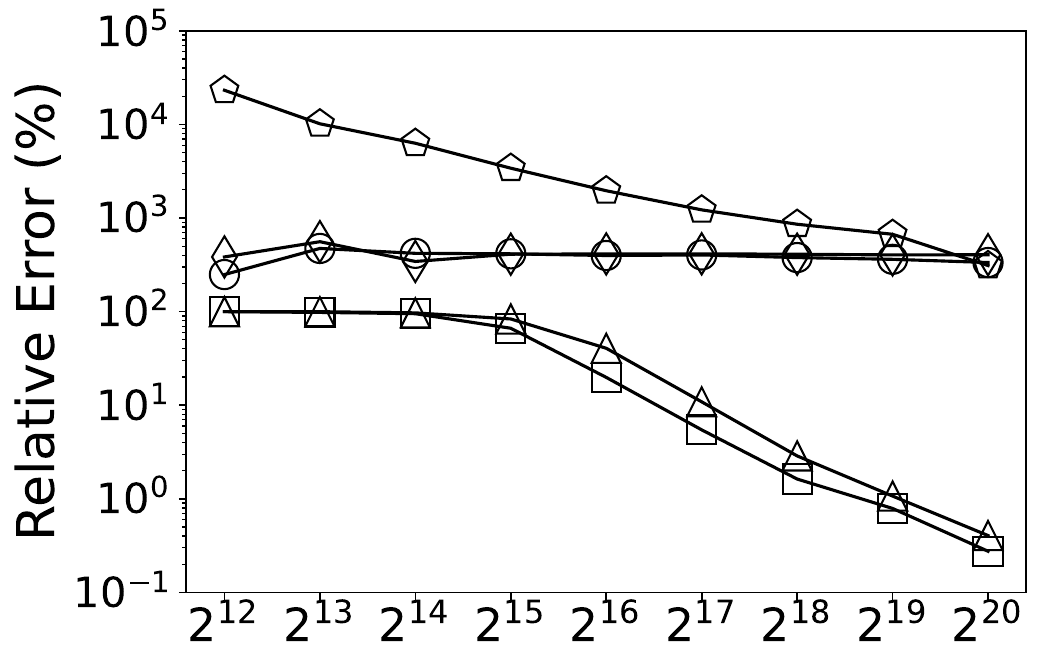}
        \subcaption{Yahoo-songs (edges)}
    \end{minipage}
    
    \begin{minipage}{0.22\textwidth}
        \centering
        \includegraphics[width=\textwidth]{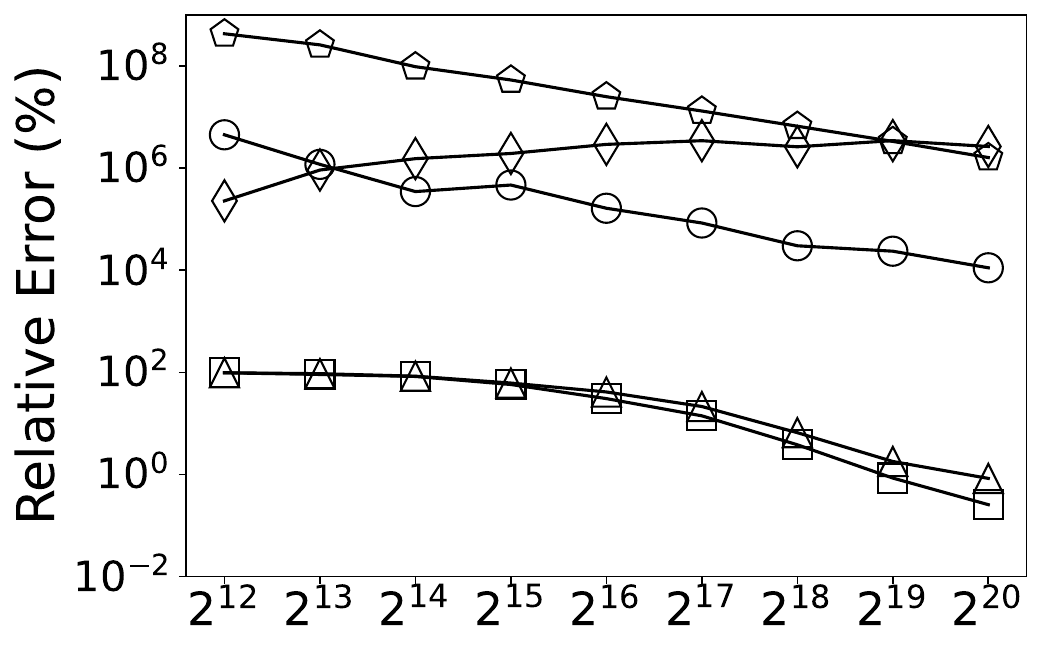}
        \subcaption{StackoverFlow (edges)}
    \end{minipage}
    \hspace{0.015\linewidth}
    \begin{minipage}{0.22\textwidth}
        \centering
        \includegraphics[width=\textwidth]{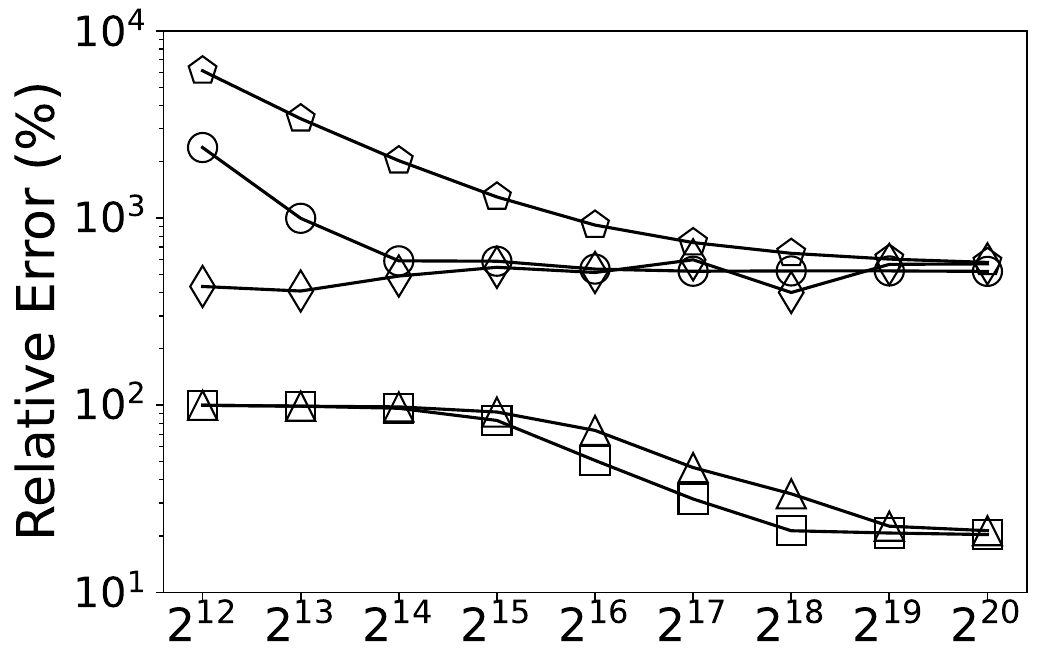}
        \subcaption{LiveJournal (edges)}
    \end{minipage}
    \hspace{0.015\linewidth}
    \begin{minipage}{0.22\textwidth}
        \centering
        \includegraphics[width=\textwidth]{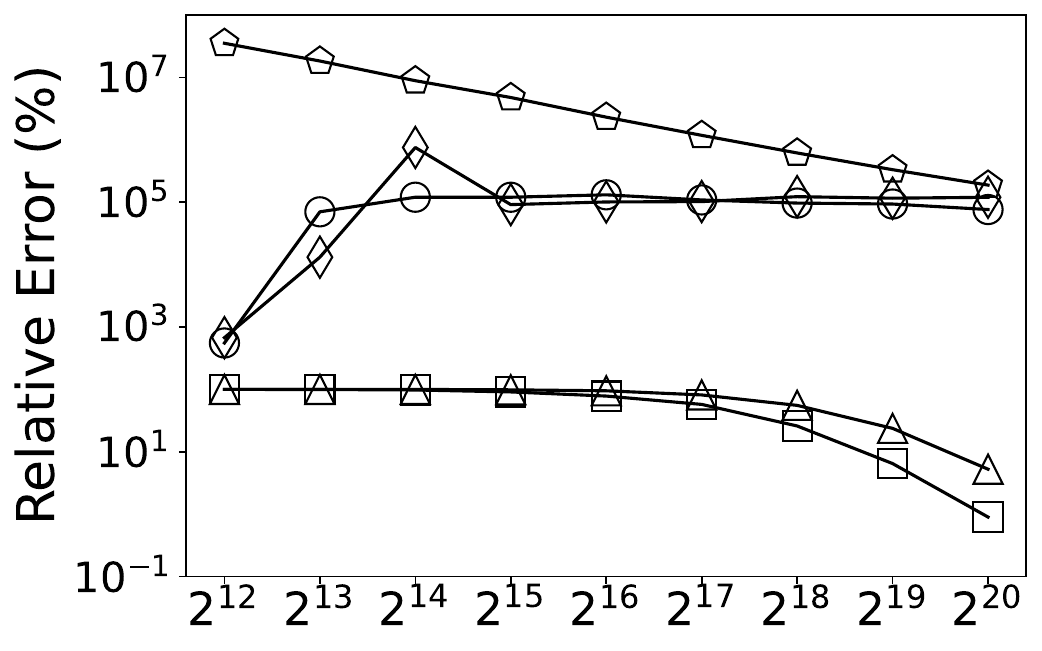}
        \subcaption{Delicious-ui (edges)}
    \end{minipage}
    \hspace{0.015\linewidth}
    \begin{minipage}{0.22\textwidth}
        \centering
        \includegraphics[width=\textwidth]{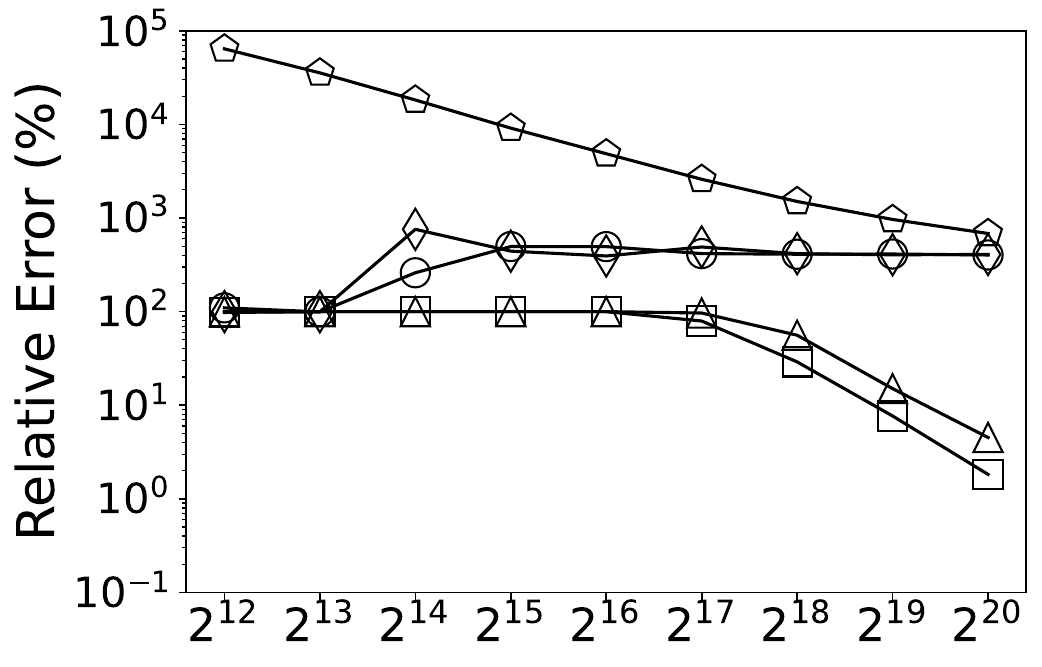}
        \subcaption{Orkut (edges)}
    \end{minipage}

    
    \caption{Relative Error under Different Sampling Sizes}
    \vspace{-0.3em}
    \label{fig:exp_accuracy}
\end{figure*}
\begin{figure*}[]
    \centering
     \vspace{-0.4em}

    \includegraphics[width=0.54\textwidth]{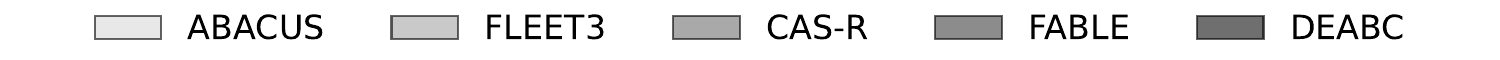}

    
    \begin{minipage}{0.48\textwidth}
        \centering
        \includegraphics[width=\textwidth]{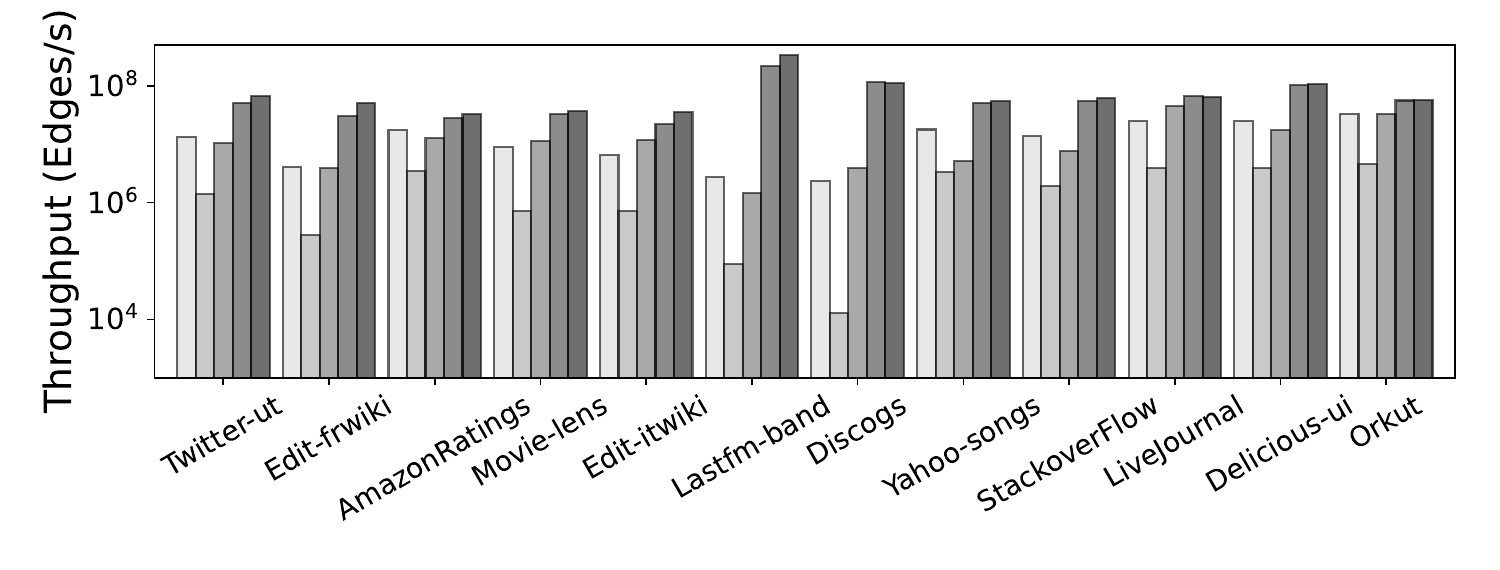}
        \vspace{-2em}
        \subcaption{Sample Size = 14}
    \end{minipage}
    \hspace{0.015\linewidth}
    \begin{minipage}{0.48\textwidth}
        \centering
        \includegraphics[width=\textwidth]{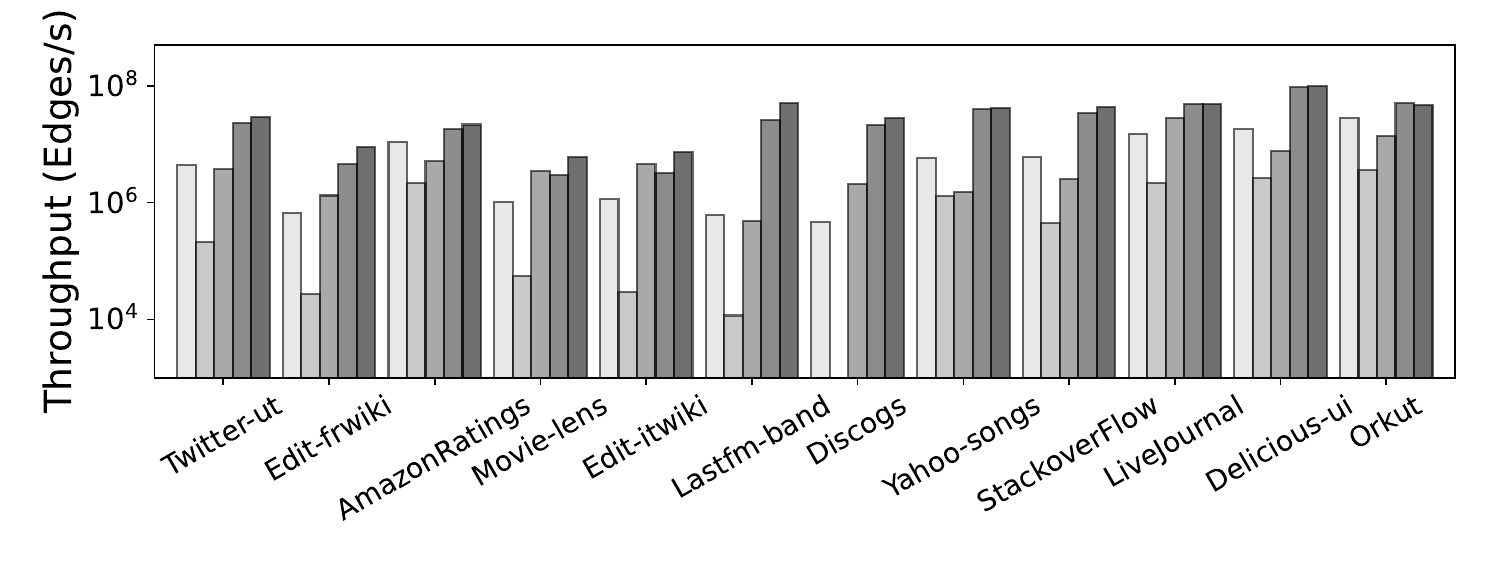}
        \vspace{-2em}
        \subcaption{Sample Size = 16}
    \end{minipage}

    \begin{minipage}{0.48\textwidth}
        \centering
        \includegraphics[width=\textwidth]{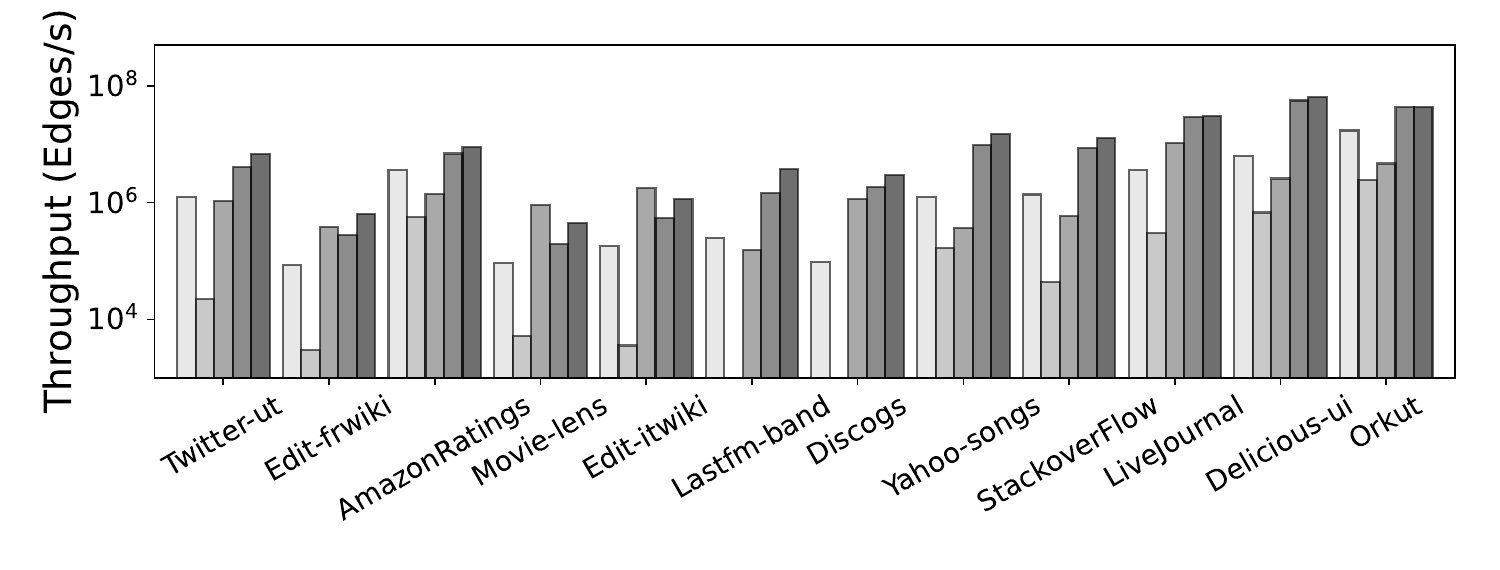}
        \vspace{-2em}
        \subcaption{Sample Size = 18}
    \end{minipage}
    \hspace{0.015\linewidth}
    \begin{minipage}{0.48\textwidth}
        \centering
        \includegraphics[width=\textwidth]{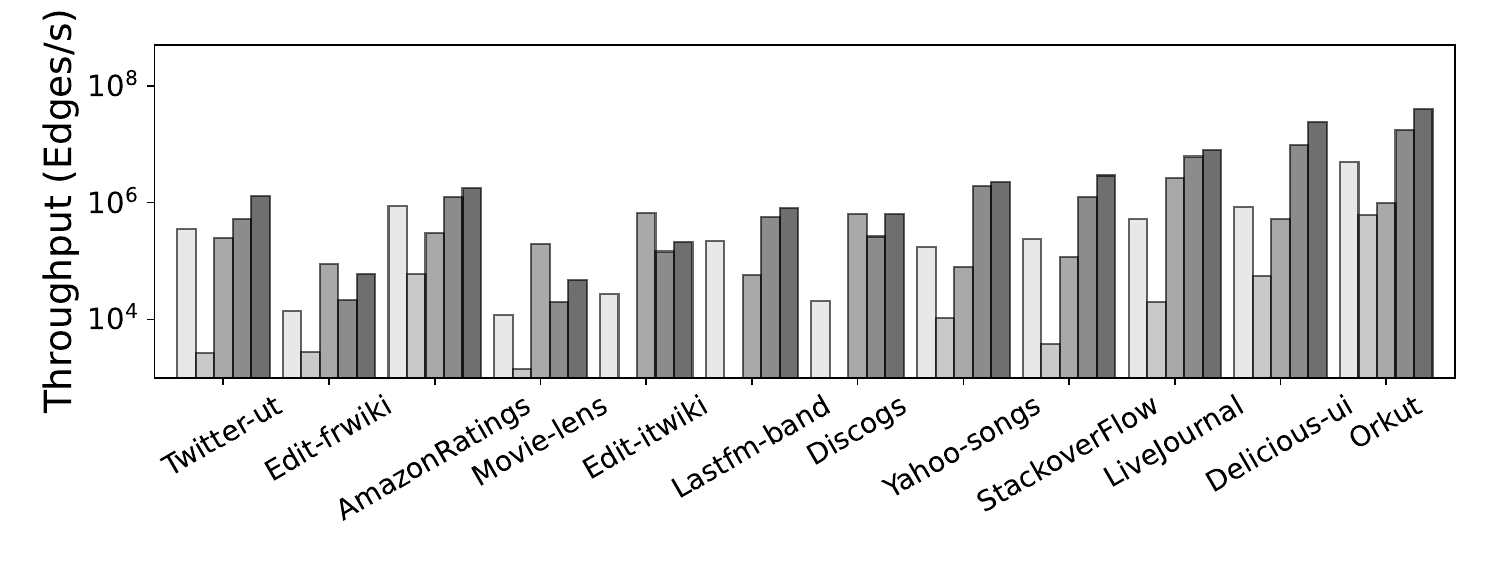}
        \vspace{-2em}
        \subcaption{Sample Size = 20}
    \end{minipage}
    
    \caption{Throughput over All Datasets}
    \label{fig:throughput_dataset}
    \vspace{-0.9em}
\end{figure*}


This section presents our experimental results. All algorithms are implemented in C++ and compiled using GNU GCC 4.8.5 with the \texttt{-O3} optimization level, running on an Intel(R) Xeon(R) Platinum 8373C CPU @ 2.60GHz with 16GB of main memory. The time cost is measured by the elapsed wall-clock time during execution.
\revise{The source code is publicly available at \url{https://github.com/Lingkai981/deabc}}

\stitle{Datasets.} We use 12 public datasets, as shown in the Table \ref{tab:syn_data}. \stackoverflow is from SNAP~\footnote{\url{http://snap.stanford.edu/data/index.html}}, \editfrwiki and \edititwiki are from Network Repository (NR)~\footnote{\url{http://networkrepository.com/}}, and the rest are from KNOECT~\footnote{\url{http://konect.uni-koblenz.de/}}. 
Among these, \twitterut, \editfrwiki, \amazonratings, \edititwiki, \lastfmband, \discogs, \stackoverflow, and \deliciousui are real-world streaming bipartite graphs with duplicate edges, which are suitable for evaluating the performance of algorithms in streaming scenarios.
To evaluate the impact of edge duplication, we used four bipartite graphs without duplicates, \movielens, \yahoosongs, \livejournal, and \orkut. For each datasets, edge repetitions follow a $Geometric\left(\frac{1}{1+\lambda}\right)$ distribution with duplication ratio $\lambda$, and the stream is randomly shuffled. The default duplication rate is 50\%.

\begin{figure}
    \centering
    \includegraphics[width=0.5\textwidth]{figure/accuracy/de_accuracy_legend_only.pdf}

    
    
    
    \begin{minipage}{0.22\textwidth}
        \centering
        \includegraphics[width=\textwidth]{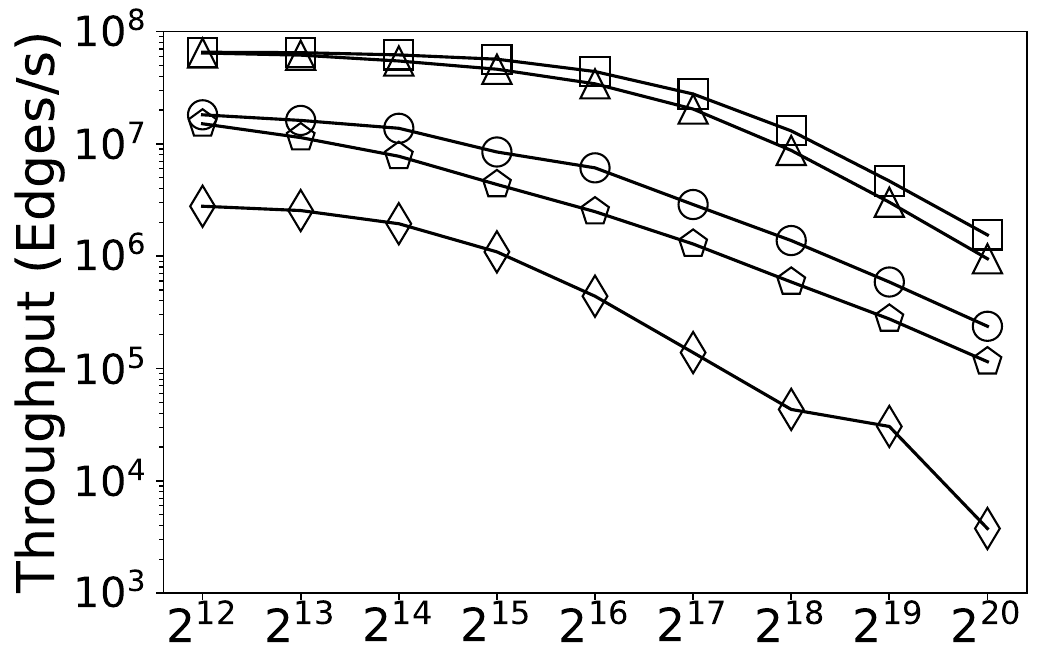}
        \subcaption{StackoverFlow (edges)}
    \end{minipage}
    \hspace{0.015\linewidth}
    \begin{minipage}{0.22\textwidth}
        \centering
        \includegraphics[width=\textwidth]{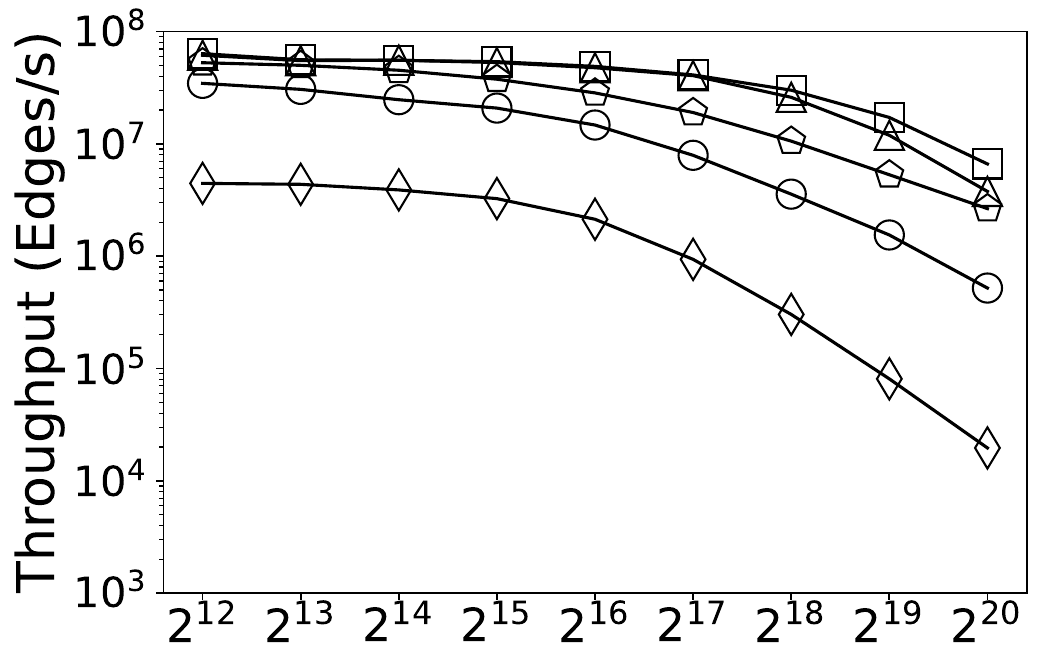}
        \subcaption{LiveJournal (edges)}
    \end{minipage}
    
    \begin{minipage}{0.22\textwidth}
        \centering
        \includegraphics[width=\textwidth]{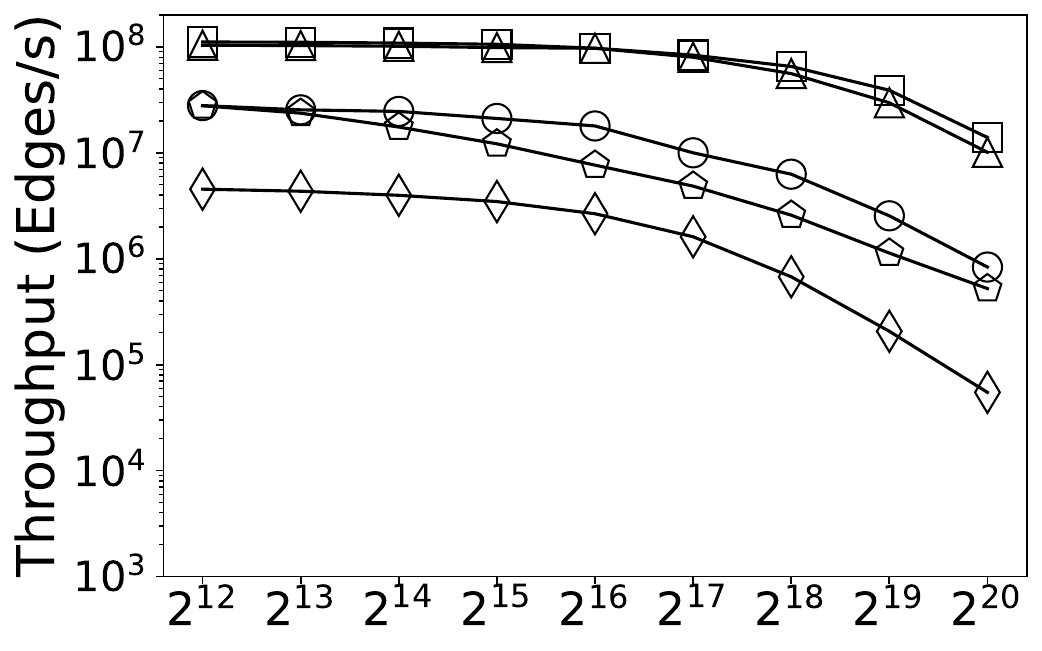}
        \subcaption{Delicious-ui (edges)}
    \end{minipage}
    \hspace{0.015\linewidth}
    \begin{minipage}{0.22\textwidth}
        \centering
        \includegraphics[width=\textwidth]{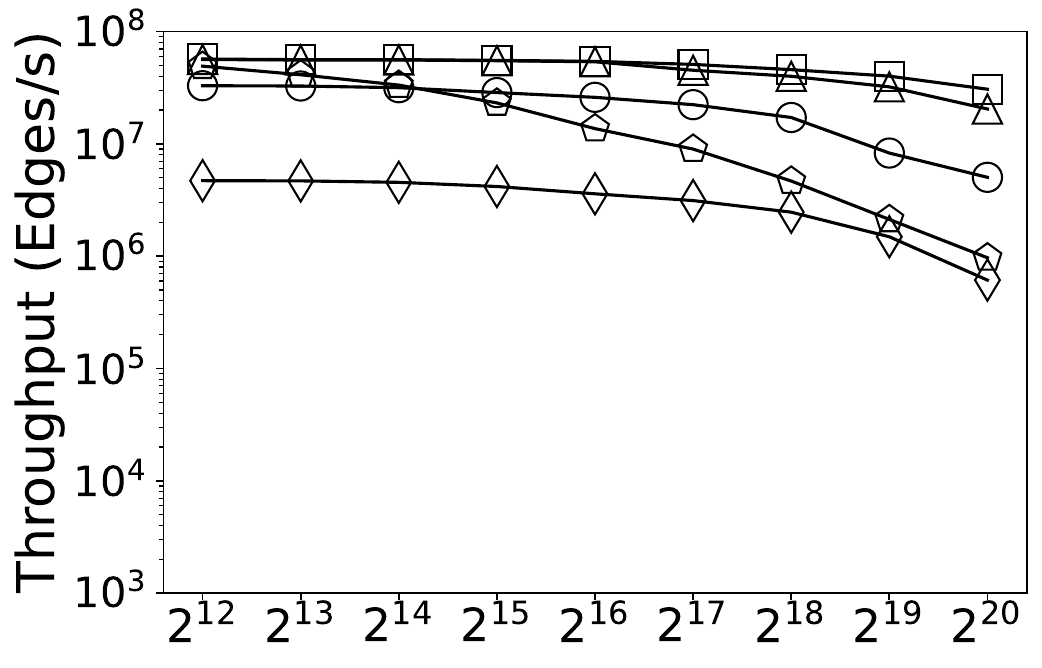}
        \subcaption{Orkut (edges)}
    \end{minipage}

    \caption{Throughput under Different Sampling Sizes}
    \vspace{-0.8em}
    \label{fig:throughput_size}
\end{figure}
\begin{figure*}
    \centering
    \vspace{-0.96em}

    \includegraphics[width=0.5\textwidth]
    {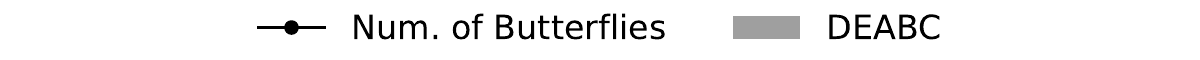}

    
    \begin{minipage}{0.96\textwidth}
        \centering
        \includegraphics[width=\textwidth]{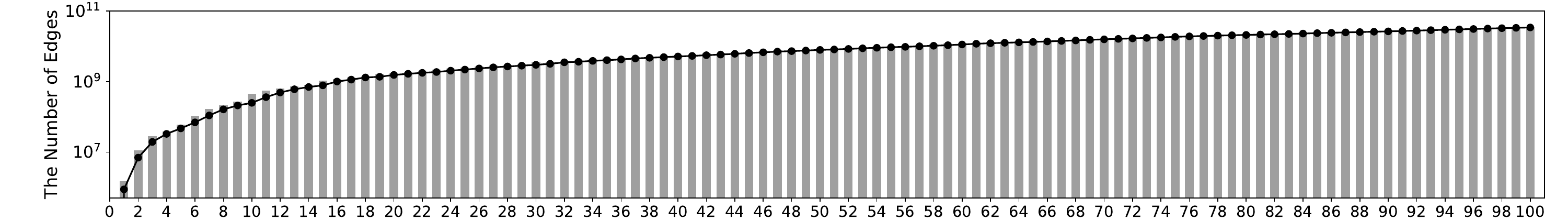}
        \subcaption{Edit-frwiki (time snapshots)}
    \end{minipage}
    
    \begin{minipage}{0.96\textwidth}
        \centering
        \includegraphics[width=\textwidth]{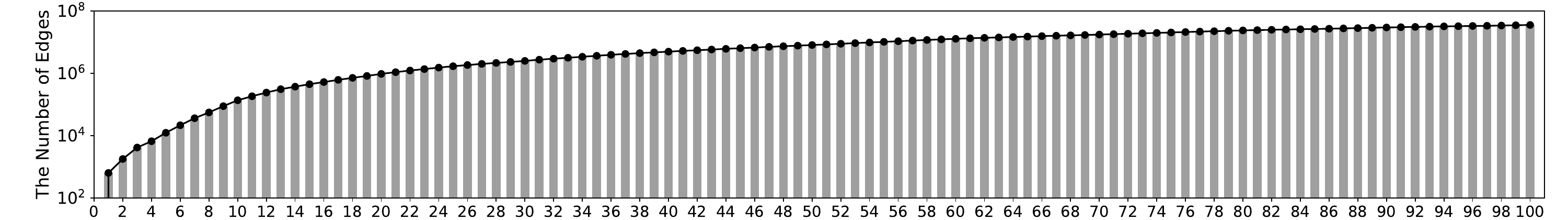}
        \subcaption{AmazonRatings (time snapshots)}
    \end{minipage}
    
    
    
    
    
    
    \begin{minipage}{0.96\textwidth}
        \centering
        \includegraphics[width=\textwidth]{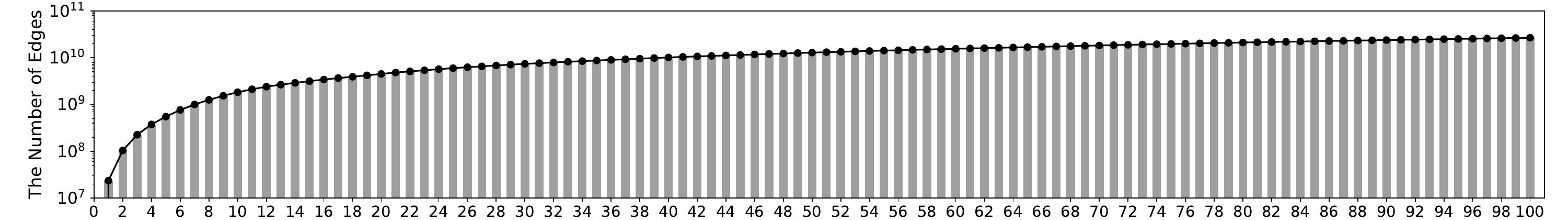}
        \subcaption{StackoverFlow (time snapshots)}
    \end{minipage}
    
    
    


    \caption{The Estimated Number of Butterflies over Time}
    \label{fig:bc_over_time}
    \vspace{-0.96em}
\end{figure*}

\begin{figure*}[]
    \centering
    

    \includegraphics[width=0.56\textwidth]{figure/Throughput/plot_five_bar_charts_legend.pdf}
    
    \begin{minipage}{0.48\textwidth}
        \centering
        \includegraphics[width=\textwidth]{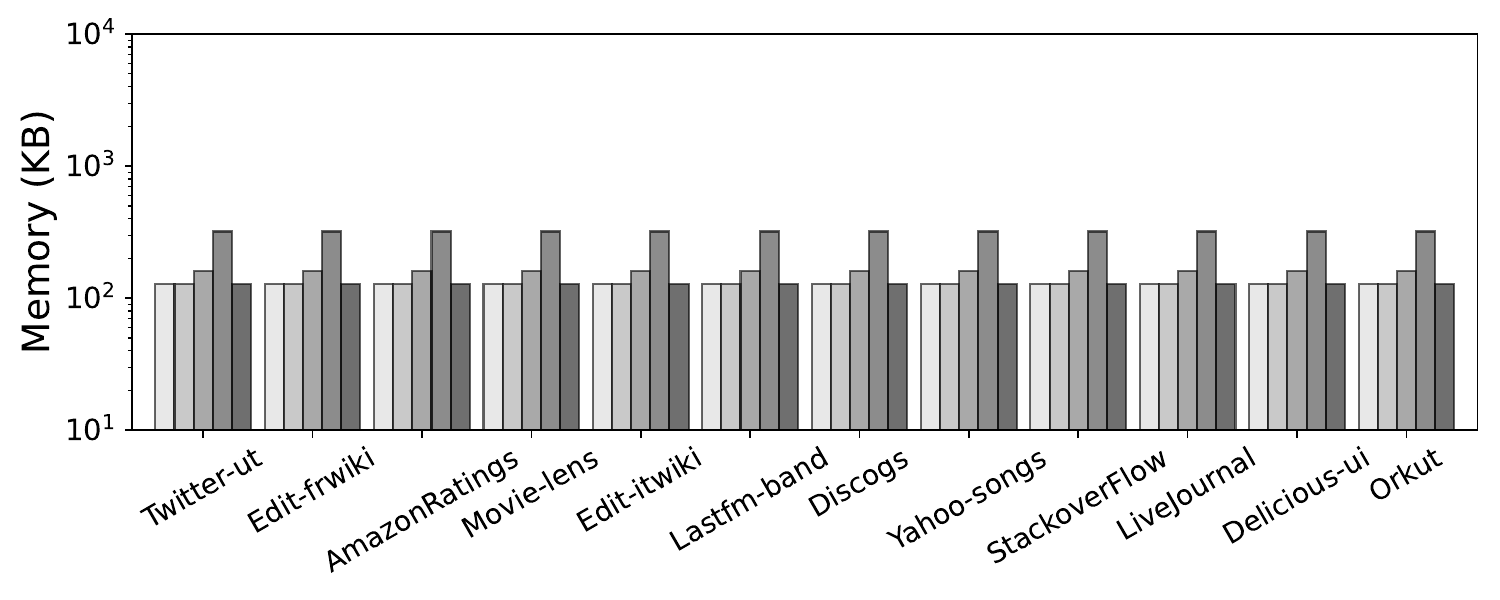}
        \vspace{-2em}
        \subcaption{Sample Size = 14}
    \end{minipage}
    \hspace{0.015\linewidth}
    \begin{minipage}{0.48\textwidth}
        \centering
        \includegraphics[width=\textwidth]{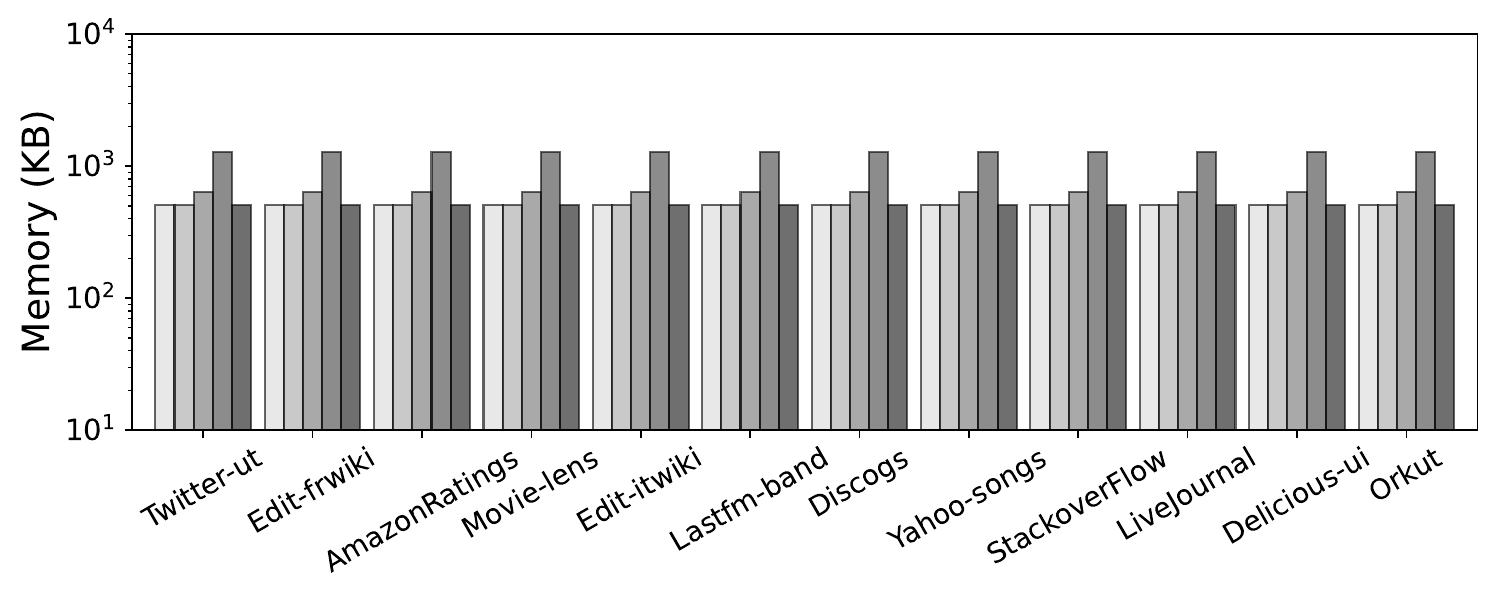}
        \vspace{-2em}
        \subcaption{Sample Size = 16}
    \end{minipage}

    \begin{minipage}{0.48\textwidth}
        \centering
        \includegraphics[width=\textwidth]{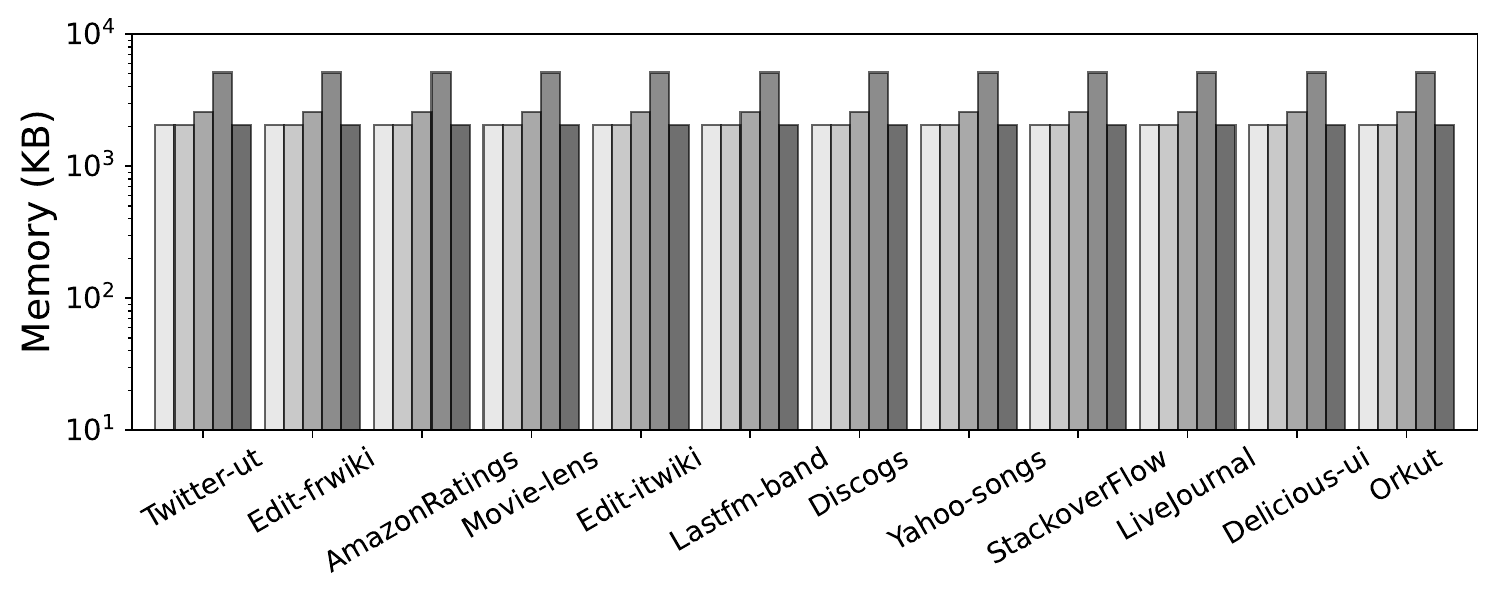}
        \vspace{-2em}
        \subcaption{Sample Size = 18}
    \end{minipage}
    \hspace{0.015\linewidth}
    \begin{minipage}{0.48\textwidth}
        \centering
        \includegraphics[width=\textwidth]{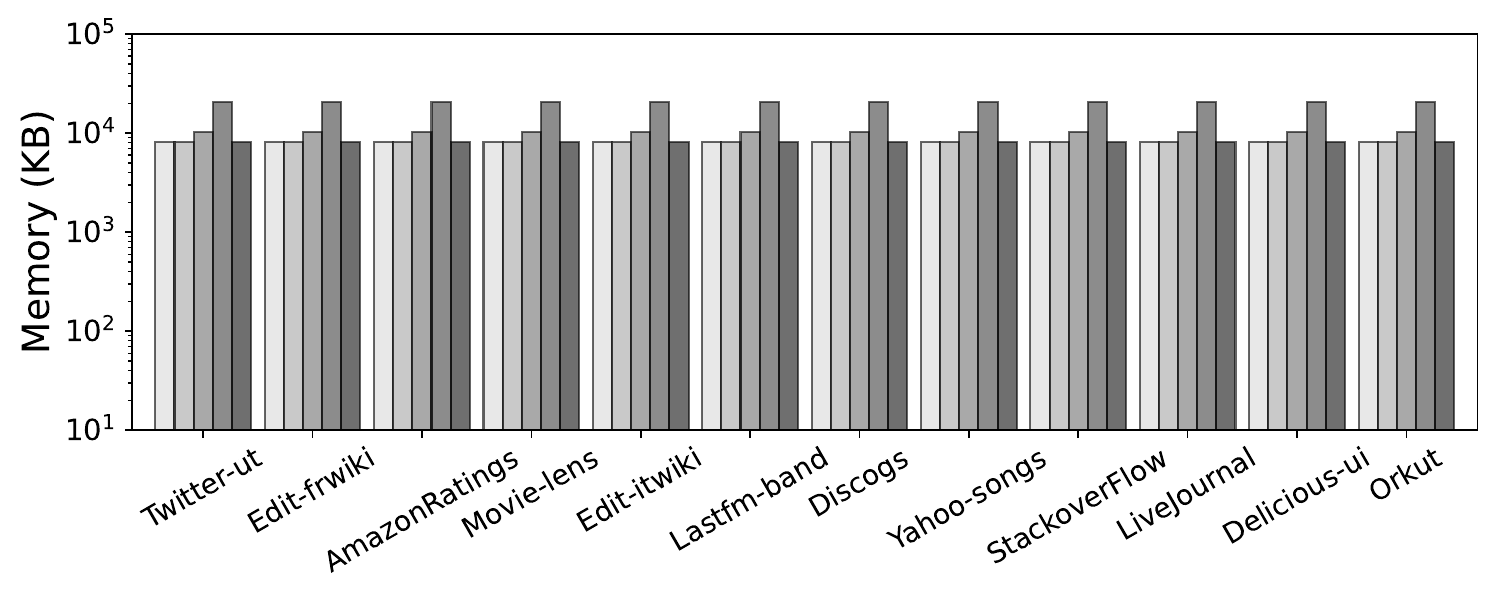}
        \vspace{-2em}
        \subcaption{Sample Size = 20}
    \end{minipage}
    

    \caption{Memory Usage over All Datasets}
    \vspace{-1.2em}
    \label{fig:memory}
\end{figure*}
\begin{figure}
    \centering
     \vspace{-0.5em}
    \includegraphics[width=0.48\textwidth]{figure/accuracy/de_accuracy_legend_only.pdf}

    \vspace{-0.5em}
    
    \begin{minipage}{0.22\textwidth}
        \centering
        \includegraphics[width=\textwidth]{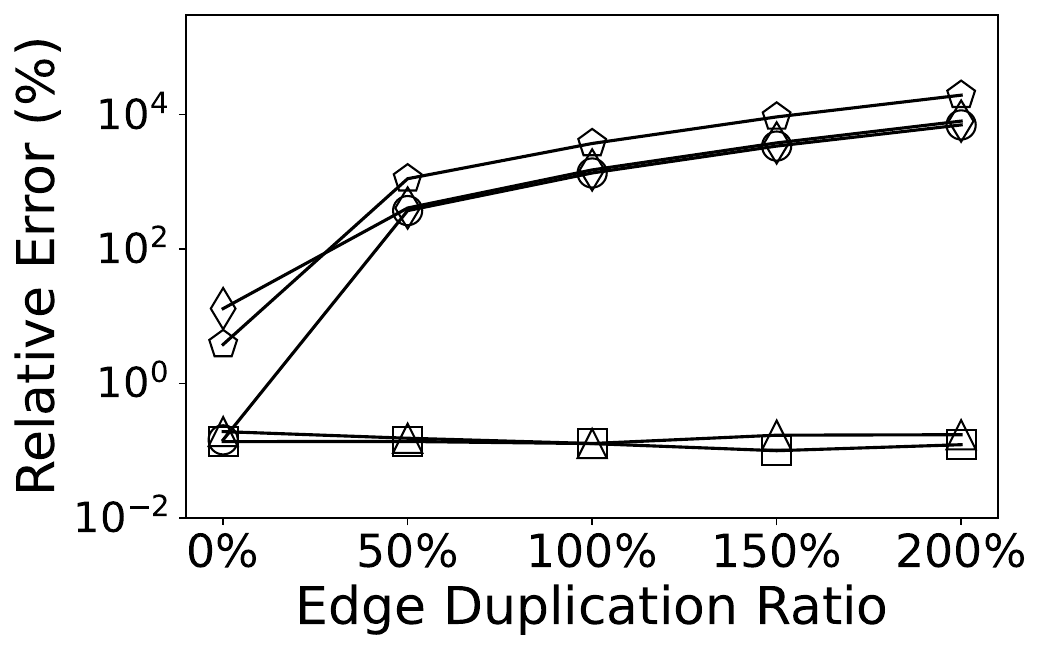}
        \subcaption{Movie-lens}
    \end{minipage}
    \hspace{0.015\linewidth}
    \begin{minipage}{0.22\textwidth}
        \centering
        \includegraphics[width=\textwidth]{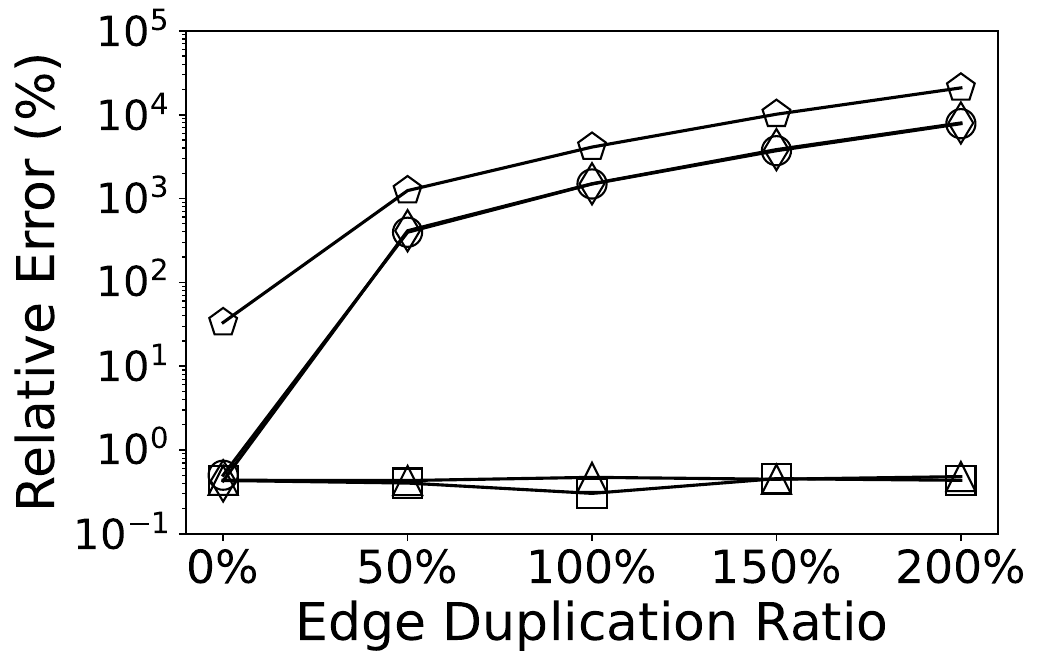}
        \subcaption{Yahoo-songs}
    \end{minipage}

    
    \begin{minipage}{0.22\textwidth}
        \centering
        \includegraphics[width=\textwidth]{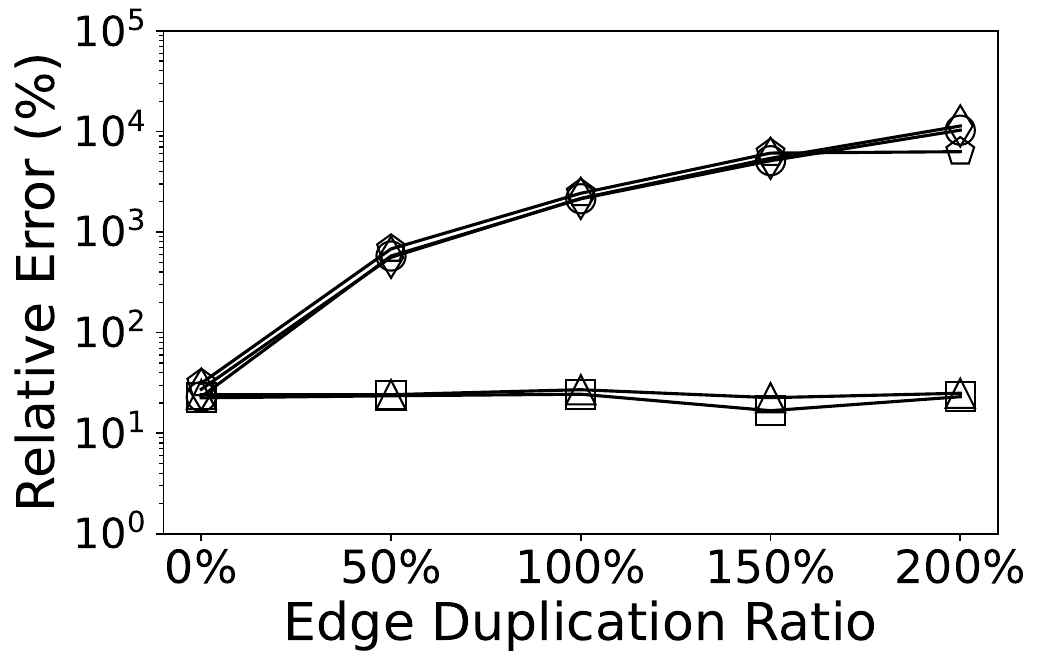}
        \subcaption{LiveJournal}
    \end{minipage}
    \hspace{0.015\linewidth}
    \begin{minipage}{0.22\textwidth}
        \centering
        \includegraphics[width=\textwidth]{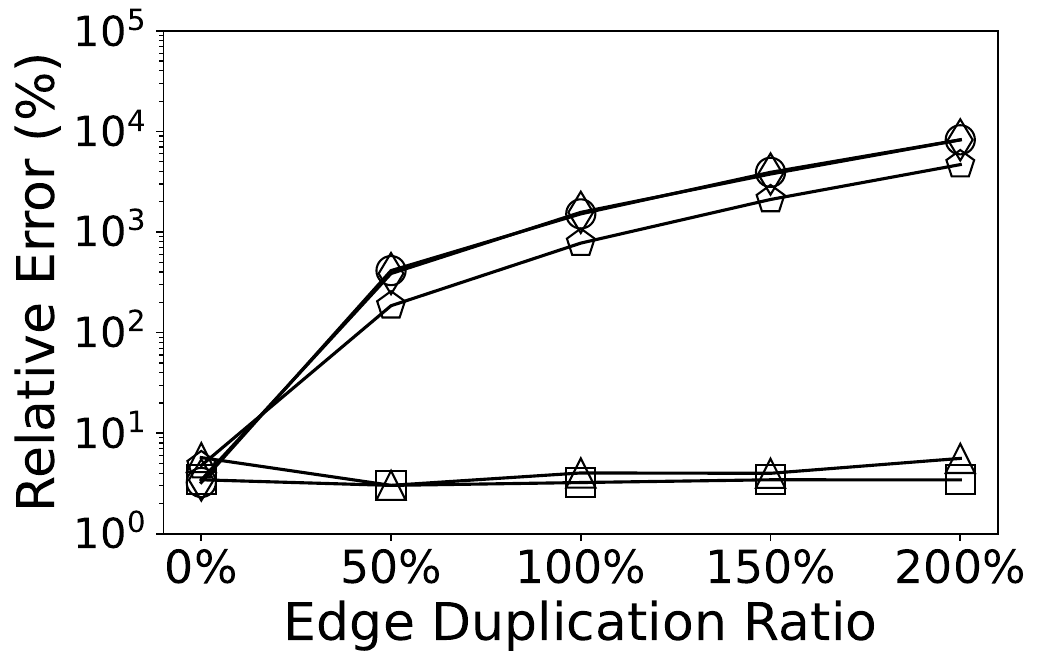}
        \subcaption{Orkut}
    \end{minipage}

  \vspace{-0.5em}
    
    \caption{Relative Error across Varying Edge Duplication Ratios (Sample Size = 18)}
    \vspace{-2em}
    \label{fig:dup}
\end{figure}


\stitle{Algorithms.} We evaluate the performance of our method \deabcpro in comparison with three state-of-the-art methods for bufferfly counting over bipartite graph streams: \abacus~\cite{DBLP:conf/icde/PapadiasKPQM24}, \fleet~\cite{DBLP:conf/cikm/Sanei-MehriZST19}, \cas~\cite{DBLP:journals/tkde/LiWJZZTYG22}, and \fable~\cite{sun2024fable}.
\begin{itemize}[leftmargin=*]

    \item \textbf{\abacus} is a streaming butterfly counting algorithm based on random sampling, which is proposed in~\cite{DBLP:conf/icde/PapadiasKPQM24}.

    \item \textbf{\fleet} is proposed in \cite{DBLP:conf/cikm/Sanei-MehriZST19}. We use the recommended optimal value of $\gamma=0.75$ as the reservoir resizing parameter.

    \item \textbf{\cas} is the best method in~\cite{DBLP:journals/tkde/LiWJZZTYG22} and we use the recommended optimal value of $\lambda=0.75$ as the ratio of memory usage of AMS sketch to total memory.

    \item \textbf{\fable} is proposed in~\cite{sun2024fable}, which uses \kmv sketch to process duplicate edges in a bipartite graph stream.
    
\end{itemize}

\stitle{Evaluation Metrics.} Our evaluation focuses on two primary metrics following~\cite{DBLP:conf/icde/PapadiasKPQM24,DBLP:conf/cikm/Sanei-MehriZST19,DBLP:journals/tkdd/SheshboloukiO22}: \textit{Relative Error} and \textit{Throughput}. \textit{Relative Error} (the lower the better) is used to measure the accuracy of the predictions by comparing the difference between the predicted and actual values relative to the actual value. Its formula is that
$
\textit{Relative Error} = \frac{|c_{\Join} - \hat{c_{\Join}}|}{c_{\Join}}  \times 100\%
$

\textit{Throughput} (the larger the better), on the other hand, measures the efficiency of the algorithm by calculating the number of edges processed per second.

\stitle{Exp-1: Accuracy.}
We evaluate the accuracy of estimating butterfly counts across all datasets. The relative error was computed for different sample sizes, which ranged from $2^{12}$ to $2^{20}$ edges. For convenience, we represent the sample size using powers of $2$. The experiments were repeated 100 times and the average was calculated. The results are shown in Figure~\ref{fig:exp_accuracy}.

Across all datasets, our proposed algorithm, \deabcpro, consistently performs better than all baseline algorithms (\abacus, \fleet, \cas, and \fable) in terms of relative error. In several datasets, the relative error of our algorithm drops below 0.1\%. Notably, in most datasets, our approach achieves a 2-4 orders of magnitude improvement compared to those that do not consider duplicate edges, underlining the critical impact of duplicate edges on butterfly counting accuracy. For example, in \amazonratings, where the duplicate edge rate is only 1\%, the relative error of our algorithm is approximately one-tenth of that of the baseline algorithms.

Compared to \fable, \deabcpro performs better in all datasets and sample size settings, especially when the sample space is small, which is consistent with the results of our variance analysis. This highlights \deabcpro’s effectiveness and accuracy, even under memory constraints. When the sampling space is particularly small, it becomes challenging to sample any butterflies, causing the relative error of both algorithms to approach 100\%.

\stitle{Exp-2: Throughput.}
We evaluate the throughput performance across all datasets for different sample sizes, specifically $2^{14}, 2^{16}, 2^{18}$, and $2^{20}$. Additionally, we also explored how throughput changes with varying sample sizes, ranging from $2^{12}$ to $2^{20}$. Due to space constraints, we only show the results of the four largest datasets (\stackoverflow, \livejournal, \deliciousui, and \orkut), and similar trends were observed across all datasets. The throughput results are shown in Figures~\ref{fig:throughput_dataset} and Figures~\ref{fig:throughput_size}.


Figure~\ref{fig:throughput_dataset} shows the throughput performance across all datasets. We can observe that \deabcpro maintains high throughput across most datasets and often outperforms the other baseline algorithms. Compared to \fable, our \deabcpro algorithm consistently maintains superior throughput performance. As the sampling space $M$ increases, this advantage expands, which is consistent with the results of our time complexity analysis.


From Figures~\ref{fig:throughput_size}, we can observe that the throughput of all algorithms decreases as the sample size increases, with \deabcpro performing the best. 
Additionally, as the sample size grows, the rate of performance decline for our \deabcpro is slower compared to other baseline algorithms, demonstrating the strong scalability of our algorithms.


\stitle{Exp-3: The Estimated Number of Butterflies over Time.}
In this experiment, we analyzed the performance of \deabcpro over time. Due to space constraints, we only show three datasets, \editfrwiki, \amazonratings, and \stackoverflow, which are real-world datasets with duplicate edges. For each dataset, we uniformly set $100$ time snapshots, with an equal number of edges processed between consecutive snapshots. The results are shown in Figure~\ref{fig:bc_over_time}. 

Our \deabcpro maintains low error at each time snapshot. As the number of edges in the data stream increases, the estimated number of butterflies also increases, and the results of both algorithms across different time snapshots are very close to the true number of butterflies, making them well-suited for real-time graph processing tasks.

\stitle{Exp-4: Memory Usage.}
We evaluate the memory usage of the five algorithms (\abacus, \fleet, \cas, \fable, and \deabcpro) across all datasets. The evaluation was conducted using four different sample sizes: $2^{14}$, $2^{16}$, $2^{18}$, and $2^{20}$. The memory usage of each algorithm is measured in kilobytes (KB), and the results, shown in Figure~\ref{fig:memory}, are presented on a logarithmic scale for better visualization of the differences across datasets.

The memory usage of \abacus, \fleet, and \deabcpro is similar and increases linearly with the sample size. For sample sizes of $2^{14}$, $2^{16}$, $2^{18}$, and $2^{20}$, they consume approximately 128KB, 512KB, 2048KB, and 8192KB, respectively. This is because they only need to store the sampled edges, and their memory usage is independent of the dataset size; even when processing very large datasets, only the memory space for the sample size is required.
The \cas algorithm requires additional data structures, resulting in approximately 25\% higher memory usage compared to \abacus, \fleet, and \deabcpro. \fable algorithm has the highest memory consumption, about double that of other algorithms, due to the need for an additional priority queue to store edge data and priority information.



\stitle{Exp-5: The Impact of Edge Duplication.}
In this experiment, we evaluated the relative error performance of the five algorithms under varying edge duplication ratios. By setting different edge duplication ratios (0\%, 50\%, 100\%, 150\%, and 200\%), we observed the behavior of each algorithm on four datasets: \movielens, \yahoosongs, \livejournal, and \orkut, with a fixed sample size of $2^{18}$.  The results are shown in Figure~\ref{fig:dup}. 


The relative error of \abacus, \fleet, and \cas increases rapidly as the edge duplication ratio grows, especially when the duplication ratio rises from 0\% to 50\%. In contrast, the algorithms, \fable and \deabcpro, demonstrate significantly better performance, and they are almost unaffected by changes in the edge duplication ratio, while our \deabcpro still performs better than \fable. Furthermore, even when handling non-duplicate edge datasets (the duplication ratio is 0\%), our algorithm achieves comparable or even better relative error rates than the baseline algorithms.

\section{Related Work}
\label{sec:related_work}

\stitle{Butterfly Counting over Static Graphs.}
Wang et al.~\cite{DBLP:conf/bigdata/WangFC14} proposed the first exact butterfly counting algorithm using wedge enumeration to process two-hop neighbors. Sanei-Mehri et al.~\cite{DBLP:conf/kdd/Sanei-MehriST18} improved it by traversing from the vertex set with smaller squared-degree sums and further introduced randomized algorithms using local sampling and sparsification for approximate counting.
In subsequent research, Wang et al.~\cite{DBLP:journals/pvldb/WangLQZZ19} reduced wedge enumerations with vertex ordering, while Shi et al.~\cite{DBLP:books/crc/22/ShiS22} improved efficiency using a parallel framework. However, these methods target static bipartite graphs and struggle with streaming settings, where dynamic edge insertions alter vertex priorities. They also require loading the full graph into memory, limiting scalability.
The Matrix-Based Butterfly Count method, proposed by Jay A. Acosta et al.~\cite{DBLP:conf/ipps/AcostaLP22}, is efficient but suffers from unacceptable memory overhead due to storing the adjacency matrix. Although the I/O-efficient butterfly counting algorithm~\cite{DBLP:journals/pacmmod/WangLLS0023,DBLP:journals/vldb/WangLLSTZ24} addresses the memory issue, frequent I/O exchanges reduce computational efficiency, making it challenging to achieve real-time butterfly counting.
Other advancements include GPU-based methods~\cite{DBLP:journals/pvldb/XuZYLDDH22,DBLP:journals/vldb/XiaZXZYLDDHM24} that use adaptive strategies to balance GPU thread workloads. Zhou et al.\cite{DBLP:journals/vldb/ZhouWC23} explored butterfly counting in uncertain graphs, while Cai et al.\cite{DBLP:journals/pvldb/CaiKWCZLG23} focused on temporal butterfly counting in temporal bipartite graphs. Additionally, methods addressing butterfly counting with differential privacy have also been proposed~\cite{DBLP:conf/infocom/Wang00LH24,DBLP:conf/icde/HeW0LNZ24}. However, these approaches are also designed for static graphs and do not work in streaming environments.

\stitle{Bufferfly Counting over Streaming Graphs.}
Sanei-Mehri et al.~\cite{DBLP:conf/cikm/Sanei-MehriZST19} proposed \kw{FLEET}, which uses adaptive sampling to estimate butterfly counts in bipartite streams with fixed memory. Li et al.~\cite{DBLP:journals/tkde/LiWJZZTYG22} introduced Co-Affiliation Sampling (\kw{CAS}), combining sampling and sketching for accurate estimates, and further improved it with \kw{CAS}-\kw{R} by incorporating reservoir sampling.
Sheshbolouki et al.~\cite{DBLP:journals/tkdd/SheshboloukiO22} proposed \kw{sGrapp}, a window-based adaptive algorithm for butterfly counting that captures temporal patterns in real-time streams. Papadias et al.~\cite{DBLP:conf/icde/PapadiasKPQM24} applied random sampling to fully dynamic graph streams. However, these methods cannot handle duplicate edges, as they cannot detect whether an inserted edge was already sampled, resulting in incorrect sampling probabilities.

Sun et al. \cite{sun2024fable} avoid the impact of duplicate edges on butterfly counting by computing the number of distinct edges by maintaining an ordered list of edge priorities for replacement and sampling. However, compared to this approach, our method achieves lower variance, higher efficiency, and reduced memory requirements.
There are other algorithms~\cite{DBLP:journals/pacmmod/ZhangCWYG23,DBLP:journals/corr/abs-2310-11886} rely on global sampling (e.g., vertex, edge, or wedge sampling) and are unsuitable for streaming settings.

\section{Conclusion}
\label{sec:conclusion}

In this paper, we present \deabcpro, a priority-based sampling algorithm for counting butterflies in streaming bipartite graphs with duplicate edges. It handles duplicates effectively while providing unbiased estimates with low variance. We prove its theoretical guarantees and validate its performance on real-world datasets. Results show that \deabc remains efficient and accurate even under memory constraints, making it well-suited for large-scale dynamic streaming graph analysis.


\section*{Acknowledgments}
Long Yuan is supported by NSFC62472225.
Xuemin Lin is supported by NSFC U2241211 and U20B2046.

\bibliographystyle{IEEEtran}
\bibliography{sample-base}

\end{document}